\documentclass{amsart}

\usepackage{amssymb}
\usepackage{amsmath}
\usepackage{amsfonts}
\usepackage{graphicx}

\newtheorem{theorem}{Theorem}[section]
\newtheorem{lemma}[theorem]{Lemma}
\newtheorem{prop}[theorem]{Proposition}
\theoremstyle{definition}
\newtheorem{definition}[theorem]{Definition}

\theoremstyle{remark}
\newtheorem{remark}[theorem]{Remark}

\numberwithin{equation}{section}

\newcommand{\DIS}{\displaystyle}
\newcommand{\eps}{\varepsilon}
\newcommand{\fie}{\varphi}
\newcommand{\s}{{S}}
\newcommand{\cdd}{\cdot\cdot}
\newcommand{\ti}[2]{\tau_{#1\cdd#2}}
\newcommand{\tih}[3]{\tau_{#1{~\!\widehat{#2}~\!}#3}}
\newcommand{\Vi}[2]{V_{#1\cdd#2}}
\newcommand{\Vih}[3]{V_{#1{~\!\widehat{#2}~\!}#3}}
\newcommand{\bigF}[1]{\mbox{\large$\displaystyle #1$}}

\begin{document}

\title{Discretisations of constrained KP hierarchies}

\author{Ralph Willox}
\address{Graduate School of Mathematical Sciences\\
the University of Tokyo\\
3-8-1 Komaba, Meguro-ku, Tokyo 153-8914, Japan}
\email{willox@ms.u-tokyo.ac.jp}
\thanks{R.W. would like to acknowledge support from the Japan Society for the Promotion of Science, through the JSPS grant : KAKENHI 24540204. }

\author{Madoka Hattori}
\address{Nippon Life  Insurance Company, Actuarial Department\\
1-6-6 Marunouchi, Chiyoda-ku, Tokyo 100-8288 Japan}
\email{hattori34650@nissay.co.jp}



\date{}


\begin{abstract}
We present a discrete analogue of the so-called symmetry reduced or `constrained' KP hierarchy. As a result we obtain integrable discretisations, in both space and time, of some well-known continuous integrable systems such as the nonlinear Schr\"odinger equation, the Broer-Kaup equation and the Yajima-Oikawa system, together with their Lax pairs. It will be shown that these discretisations also give rise to a discrete description of the entire hierarchy of associated integrable systems. The discretisations of the Broer-Kaup equation and of the Yajima-Oikawa system are thought to be new. 
\end{abstract}

\maketitle
\section{Introduction}\label{introduction}
There has been remarkable progress in the study of integrable discrete systems, since the initial discovery of integrable difference schemes for the nonlinear Schr\"odinger (NLS) equation \cite{AbloL} or the Korteweg-de Vries (KdV) equation \cite{Hiro-dKdV}, now more than 35 years ago. However, although tremendous conceptual as well as technical advances were made during the 1980's (see e.g. the remarkable series of papers of which \cite{DateIII,DateIV} constitute the middle part), interest in integrable discretisations temporarily waned during the latter half of that decade, as most research activity focused on purely analytical approaches to integrability, better suited to continuous systems. All this changed dramatically however during the following decade, and this mainly due to a series of simultaneous but at first seemingly unrelated discoveries. First there was the (re-)discovery of a discrete form of the Painlev\'e I equation in the context of a field-theoretical model of 2-dimensional gravity \cite{Brezin}, which led to an explosion of results on discrete Painlev\'e equations, and to the development of integrability tests for discrete mappings \cite{GramPRL} (cf. \cite{GramCIMPA} for a review of the history and the mathematical particulars of the field). Secondly, there was the discovery of the first solitonic cellular automaton \cite{TakahashiS}, followed by the discovery that this system is intimately related to the discrete KdV equation through a special limiting procedure, the ultradiscrete limit \cite{Toki96}. These results and the subsequent realization that systems obtained through `ultradiscretisation' are in fact related to certain exactly solvable lattice models (at their crystal limits), spawned an enormous amount of research activity in both quantum integrable as well as classically integrable cellular automata. However, over the years, a strange situation has developed. Whereas considerable advances have been made in the field of integrable automata, or tropical integrable systems as they are also known, and whereas numerous examples of such systems associated to a large variety of symmetry algebras have been constructed, the 1+1 dimensional, classical, discrete integrable systems that ought to contain the tropical ones at the ultradiscrete limit are still largely unknown (recent advances in the theory of geometric crystals and their relation to Yang-Baxter maps notwithstanding). Similarly, although the relationship between continuous Painlev\'e equations and 1+1 dimensional integrable systems through similarity reduction has become text-book material, the situation in the discrete case is quite different. The structure and properties of discrete Painlev\'e equations are, by now, of course quite well-understood. But the precise relationship these systems bear to discrete 1+1 dimensional (lattice) equations is still very much an open problem. An exception is the case of {\em multiplicative} discretisations (or q-analogues) of continuous integrable systems, in which case not only numerous examples of reductions of q-lattice equations to q-Painlev\'e equations are known, but for which there is also the beginning of a general theory \cite{Haine,Kaji-qP,KakeiK, Lin, Suzuki}. The fact that q-lattice equations allow for a simple notion of `similarity', analogous to the continuous one, is of course partly responsible for this success. It is worth pointing out that a discrete equivalent of a similarity reduction is much harder to define in the case of {\em additive} discrete equations. Such a notion does exist for certain types of  such equations \cite{Nijhoff-PII,Nijhoff-PVI}, but it is very much tied up with the method that is used to generate the original discrete systems -- discretisations of the Gelfand-Dikii hierarchy \cite{Gelfand}, formulated in terms of linear integral equations that arise from an infinite-matrix scheme \cite{Nijhoff-LGD} -- and it is therefore difficult to implement in other cases. 

Which brings us to the third major development that took place during the 1990's: the realization that there is a deep relation between integrable quantum field theories and crucial elements of the theory of classical discrete integrable systems \cite{Kricheveretal}, and the ensuing resurgence of interest in discretisations of classical integrable systems. However, whereas q-deformed versions of the Gelfand-Dikii hierarchies -- i.e. of systems with underlying A-type (affine) symmetry algebras -- are rather easily obtained,  obtaining {\em additive} discretisations is a much harder problem and obtaining a discrete analogue of the Drinfeld-Sokolov construction \cite{Drinfeld} for general symmetry algebras turns out to be even more difficult. In fact, despite recent interesting results for Toda field equations related to certain types of affine Lie algebras \cite{Gari}, one cannot escape the impression that the situation has not changed all that much since the middle of the 1990's, when the problem of constructing discrete integrable systems related to general affine Lie algebras was originally raised \cite{Ward}. This problematic situation is compounded by the fact that, even if for a certain affine Lie algebra a discrete integrable system is known, it almost always corresponds to what is called the principal realization of that algebra. Examples of discrete integrable systems related to other, e.g. homogeneous \cite{FrenkelK, IkedaY}, realizations of affine Lie algebras are exceedingly rare. On the other hand, given an integrable lattice system, to date, there is no known procedure for deducing the associated symmetry algebra and for several well-known discrete integrable systems it is in fact not clear at all what affine Lie algebra they are related to.

The results presented in this paper will hopefully shed some new light on the problematics sketched above. Our main aim is to investigate a discrete analogue of a technique which can be used to obtain 1+1 dimensional integrable systems, associated to non-principal realizations of A-type affine Lie algebras, through dimensional reduction of the KP hierarchy: the so-called symmetry-constraint or symmetry reduction technique.  This technique was introduced in \cite{SidorenkoS} and was developed further in \cite{KonopelchenkoS} and \cite{Cheng}. In a standard dimensional reduction of the KP hierarchy, a subset of the flows in the hierarchy is trivialized, for example by demanding that the solutions of the equations in the hierarchy do not depend on a certain set of variables. In the case of a symmetry constraint however, the squared eigenfunction symmetry -- which in a certain sense realizes the general action of $GL(\infty)$ on the Sato-Grassmannian \cite{Mulase,RW-IP,RW-JMP2} -- is restricted to coincide with one of the flows in the hierarchy, which results in an integrable dimensional reduction. Given this close connection between the symmetry constraint and the general symmetry group that underlies the KP theory, it is natural to implement this technique directly on the KP tau functions, rather than on the nonlinear variables. This reformulation, in fact, leads to important new insights into the nature of the solutions of these hierarchies \cite{Loris1} and it is this approach that is best suited to the discrete setting and that will be adopted here. 

The next section will be devoted to a quick overview of the necessary technical background material on symmetry constraints and on the discrete KP hierarchy. The main object we shall be concerned with throughout this paper, is the celebrated Hirota-Miwa (HM) equation \cite{HirotaM,Miwa}
\begin{multline}
(\mu - \nu)~\! \tau(l +1, m, n) \tau(l, m+1 ,n+1)~\! \\ + ~\! (\nu - \lambda)~\! \tau(l, m+1, n) \tau(l+1, m,n+1) ~\! \\ +~\!  (\lambda-\mu)~\! \tau(l, m, n+1) \tau(l+1, m+1 ,n) ~\! = ~\! 0~\!,\label{HM}
\end{multline}
where the complex valued functions $\tau(l, m,n)$ are, for now, defined on $\mathbb{Z}^3$. The parameters $\lambda, \mu ,\nu\in\mathbb{C}^\times$ are taken to be mutually distinct and their reciprocals play the role of lattice parameters in the continuum limit of \eqref{HM} with respect to the continuous variables $x_1, x_2$ and $x_3$ \cite{Miwa}:
\begin{gather}
x_1 = x_1^0+\frac{l}{\lambda}+\frac{m}{\mu}+\frac{n}{\nu}~\!,\nonumber\\
x_2 = x_2^0+\frac{l}{2\lambda^2}+\frac{m}{2\mu^2}+\frac{n}{2\nu^2}~\!,\label{contlim}\\
x_3 = x_3^0+\frac{l}{3\lambda^3}+\frac{m}{3\mu^3}+\frac{n}{3\nu^3}~\!,\nonumber
\end{gather}
(where $x_1^0, x_2^0$ and $x_3^0$ are arbitrary complex numbers). Indeed, taking the (formal) limits $|\lambda| , |\mu|, |\nu| \rightarrow \infty$ in \eqref{HM} (in any particular order) yields, at lowest order, the Hirota-bilinear form of the KP equation:
\begin{equation}
\left(4 D_{x_3} D_{x_1} - D_{x_1}^4 - 3 D_{x_2}^2 \right) \tau(x_1,x_2,x_3) \cdot \tau(x_1,x_2,x_3) = 0~\!,\label{KP-BF}
\end{equation}
where we have used the same symbol `$\tau$' to denote the continuum limit of the original, discrete, object $\tau(l,m,n)$. Since this is only a minor abuse of notation, we shall adhere to this convention in an attempt to restrict the number of different variable names that will appear in the paper. As for the Hirota $D$ operators, the reader is referred to \cite{JimboM} for their precise definition. 

The HM equation \eqref{HM} is said to be an integrable system first and foremost because it arises in a natural way in the Sato formulation of the KP hierarchy \cite{Miwa}, but also, in a more restricted sense, because it can be obtained as the compatibility condition of a set of linear equations, just as the KdV equation is obtained from its Lax pair or the NLS equation from the Zakharov-Shabat equations. This is the definition of integrability we shall adopt throughout the paper. If a particular nonlinear system, defined on a two dimensional lattice, can be obtained as the compatibility condition of a system of {\em linear} equations involving a non-trivial spectral parameter, then we shall say that that nonlinear system is {\em integrable}.

Returning to the HM equation \eqref{HM}, it should be noted that there is of course only one true degree of freedom in the coefficients of the equation, but that this slightly contrived way  of writing the equation is particularly well-suited to taking the continuum limit. Moreover, this particular parametrization has the added advantage that the continuum limit can be taken with similar ease on the solutions of the equation.
However, the issue of parametric freedom in the coefficients of integrable lattice equations is not an entirely innocuous one. As we shall see (and as was already pointed out in \cite{Glasgowpap}), it can happen for lower dimensional lattice equations that the continuum limit does not yield a non-trivial integrable system for generic values of the coefficients in the equation and that the need arises to impose conditions on those coefficients. The original discrete system, without any restrictions on its coefficients, is nonetheless an integrable system (in the sense explained above) and we shall adopt the convention of referring to such a system as a {\em discrete analogue} of the continuous integrable system that can be obtained from it, at the continuum limit, by suitably restricting the freedom in the coefficients. What we shall call a {\em discretisation} of a continuous integrable system is a lattice equation for which the continuum limit can be carried out without any further restrictions. 

The discrete NLS, Broer-Kaup and Yajima-Oikawa equations we shall derive are integrable discretisations of their continuous namesakes in the above strict sense. These equations, together with the symmetry-constraint method used to obtain them from the HM equation \eqref{HM} will be presented in section 3, which is devoted to the main results of the paper. These include the Lax pairs for the aforementioned systems, but also the interpretation of these systems as discrete forms of the entire associated hierarchy of integrable systems, as explained in section 5. The proofs of the results presented in section 3 make up most of the remainder of the paper. In the final section we shall explain some of the problems one encounters when one tries to extend the reduction method to higher order lattice equations and we shall present a discretisation of a Melnikov system \cite{Melnikov}, i.e. a Boussinesq equation with sources, which we believe to be integrable but for which so far we have been unable to obtain a Lax pair.

\section{Preliminaries}\label{preliminaries}
We define functions $\tau : \mathbb{Z}^r \rightarrow \mathbb{C}$ on an $r$ dimensional lattice, for arbitrary $r\in\mathbb{N}$, characterized by an $r$-tuple $(a_1, a_2, \hdots, a_r)$ of mutually distinct complex numbers $a_j\in\mathbb{C}^\times$, each associated (exclusively) to a particular direction on the lattice. Locally, we shall require $\tau(\ell)$,  $\ell\in\mathbb{Z}^r$, to satisfy the HM-equation \eqref{HM} for every possible triple $(l,m,n)$ of distinct directions on the lattice. The parameters $\lambda,\mu$ and $\nu$ that appear in the equation represent the lattice parameters $a_j$ associated to the $l,m$ and $n$ directions, respectively. In order to make the notation a little lighter, especially in longer formulae, we shall denote a shift in a particular direction as a subscript to the symbol for the function on which that shift acts, and we shall only explicitly state the overall variable dependence of the functions when there is a risk of ambiguity. The HM equation \eqref{HM} shall therefore, from now on,  be written as
\begin{equation}
(\mu - \nu)~\! \tau_l \tau_{mn}~\! + ~\! (\nu - \lambda)~\! \tau_m \tau_{ln} ~\!  +~\!  (\lambda-\mu)~\! \tau_n \tau_{lm} ~\! = ~\! 0~\!.\label{HMs}
\end{equation}
The limit $r\rightarrow\infty$ of this construction defines the {\em discrete KP hierarchy} and its tau functions.
\begin{definition}\label{dKPdef}
The discrete KP hierarchy is the set of all possible HM-equations \eqref{HMs} on the infinite dimensional lattice that is obtained as $r\rightarrow\infty$.
\end{definition}
A tau function will then be defined as a function that satisfies this infinite set of HM-equations. However, in this definition, it is preferable to identify all tau functions that only differ by trivial symmetries: it is clear that if $\tau(\ell)$ satisfies \eqref{HMs}, then so does $c_0 c_1^l c_2^m c_3^n~\! \tau(\ell) $, for arbitrary $c_j\in\mathbb{C}^\times (j=0,\hdots, 3)$, which obviously defines an equivalence relation on the solution space of the HM-equations. 
\begin{definition}\label{taudef}
A tau function for the discrete KP hierarchy is a function $\tau$ in the quotient set of the solution space of all HM-equations in the hierarchy, by the equivalence relation
$$^\forall ~c_j\in\mathbb{C}^\times~ (j=0,\hdots, 3) :\quad c_0 c_1^l c_2^m c_3^n~\! \tau(\ell)~\! \sim~\! \tau(\ell)~\!,$$
for any triple $(l,m,n)$ of distinct directions on the lattice.
\end{definition}

It should be noted that this definition of the discrete KP hierarchy is different from the usual one, formulated in terms of quadratic difference equations of increasing order (cf. \cite{Ohta-dKP} or the appendix \ref{app} in this paper for further details). Definition \ref{dKPdef} however is more elementary than the standard one and it can be easily seen that it implies the latter. This fact, although fundamental, does not seem to be widely appreciated in the integrable systems community and the proof of this statement is therefore included in the appendix \ref{app}.

The discrete KP hierarchy is obtained as the compatibility condition of an infinite set of linear equations, satisfied by the so-called {\em eigenfunctions} $\psi(\ell)$, associated to a tau function $\tau(\ell)$, which are complex valued functions defined on the  same lattice as $\tau$. As the discrete KP hierarchy consists of infinitely many copies of \eqref{HMs}, it suffices to give the linear system associated to that particular equation: 
\begin{gather}
\psi_{lm} ~\!=~\! \frac{1}{\lambda-\mu}~\! \frac{\tau_l \tau_m}{\tau \tau_{lm}}~\!\left[~\!\lambda~\!\psi_m - \mu~\!\psi_l~\!\right]~\!,\nonumber\\
\psi_{mn} ~\!= ~\!\frac{1}{\mu-\nu}~\! \frac{\tau_m \tau_n}{\tau \tau_{mn}}~\!\left[~\!\mu~\!\psi_n - \nu~\!\psi_m~\!\right]~\!,\label{dKP-LP}\\
\psi_{ln} ~\!=~\! \frac{1}{\lambda-\nu}~\! \frac{\tau_l \tau_n}{\tau \tau_{ln}}~\!\left[~\!\lambda~\!\psi_n - \nu~\!\psi_l~\!\right]~\!,\nonumber
\end{gather}
the full set being obtained by repetition of the same basic pattern in all other directions. It is important to stress that the HM-equation \eqref{HMs} is only obtained as the compatibility condition of all 3 equations in \eqref{dKP-LP}, and not of just two of them. The third equation is required to fix the excess (gauge) freedom that would otherwise leave one of the coefficients in the HM equation undetermined. On the other hand, any two equations in \eqref{dKP-LP}, taken together with \eqref{HMs}, imply the third.

There also exists an alternative linear system, which plays a similar role as \eqref{dKP-LP}:
\begin{gather}
\psi^* ~\!= ~\!\frac{1}{\lambda-\mu}~\! \frac{\tau_l \tau_m}{\tau \tau_{lm}}~\!\left[~\!\lambda~\!\psi_l^* - \mu~\!\psi_m^*~\!\right]~\!,\nonumber\\
\psi^* ~\!= ~\!\frac{1}{\mu-\nu}~\! \frac{\tau_m \tau_n}{\tau \tau_{mn}}~\!\left[~\!\mu~\!\psi_m^* - \nu~\!\psi_n^*~\!\right]~\!,\label{dKP-aLP}\\
\psi^* ~\!=~\! \frac{1}{\lambda-\nu}~\! \frac{\tau_l \tau_n}{\tau \tau_{ln}}~\!\left[~\!\lambda~\!\psi_l^* - \nu~\!\psi_n^*~\!\right]~\!.\nonumber
\end{gather}
These relations arise because of the invariance of the HM equation \eqref{HMs} under point reflections $\ell \mapsto -\ell$ with respect to the origin of the lattice, and are therefore said to constitute the {\em adjoint} linear system for the discrete KP hierarchy. Its solutions $\psi^*(\ell)$ are called the adjoint eigenfunctions associated to $\tau(\ell)$. Both sets of linear equations will play a crucial role in what follows.

The discrete KP hierarchy can also be obtained directly from the fermionic construction of the continuous KP hierarchy \cite{JimboM} to which it is related by the so-called Miwa-transformation \cite{Miwa},
\begin{equation}
^\forall j= 1, 2, 3, \cdots~:\quad x_j = \sum_{i=1}^\infty~\!\frac{a_i^{-j}}{j} \ell_i~\!,\label{Miwatransf}
\end{equation}
which generalizes the relations in \eqref{contlim}.
Here, $x_1, x_2, x_3, \cdots$ represent the coordinates in which the continuous hierarchy is expressed and $\ell_1, \ell_2, \ell_3,\cdots$ represent the discrete KP lattice coordinates, with associated lattice parameters $a_1, a_2, a_3,\cdots$. Most importantly, this relation implies that the discrete KP hierarchy can be associated with the $A_\infty$ Kac-Moody algebra \cite{JimboM}, just as the continuous KP hierarchy. Integrable lattice equations associated with affine Lie algebras of type $A_N^{(1)}$ \cite{Hata} can be obtained from the discrete KP hierarchy by a process called {\em reduction}. For example, imposing the periodicity condition $\tau_{lm} = \tau$ on the tau functions and requiring that $\lambda=-\mu$ in the HM equation \eqref{HMs}, one obtains the Hirota bilinear form of the discrete KdV equation
\begin{equation}
(\mu - \nu)~\! \tau_{m'} \tau_{mn}~\! + ~\! (\nu +\mu)~\! \tau_m \tau_{m' n} ~\!  -~\!  2\mu~\! \tau \tau_n ~\! = ~\! 0~\!,\label{BdKdV}
\end{equation}
where a primed subscript denotes a down-shift on the lattice, e.g.: $m' : m\mapsto m-1$. Introducing continuous coordinates through \eqref{contlim}, this equation indeed yields the Hirota form of the KdV equation at the continuous limit $|\mu|, |\nu|\rightarrow \infty$,
\begin{equation*}
\left(4 D_{x_3} D_{x_1} - D_{x_1}^4\right) \tau \cdot \tau = 0~\!,
\end{equation*}
and one can therefore claim that the 2-dimensional lattice equation \eqref{BdKdV}
which the reduced tau functions now satisfy,  should be associated with the same $A_1^{(1)}$ affine Lie algebra that underlies the continuous equation. Note that the discrete equation \eqref{BdKdV} still possesses one degree of freedom in its coefficients.

Other famous 1+1 dimensional integrable systems, the discretisations of which can be obtained through similar reductions, include the modified KdV equation, the sine-Gordon equation or the Boussinesq equation (which will appear in the final part of the paper). However, discrete analogues (or, for that matter, discretisations) of 1+1 dimensional integrable systems that arise as non-principal realizations of (A-type) affine Lie algebras cannot be obtained through this elementary reduction technique. Perhaps the most famous example of such a system is the NLS equation, which is obtained as a homogeneous realization of the $A_1^{(1)}$ algebra \cite{IkedaY, Ikedaetal}. The (fully discrete) Ablowitz-Ladik equation \cite{AbloL} is without any doubt the best-known discretisation of the NLS equation but in \cite{Suris} another, related, discretisation is presented which is in fact better suited to numerical calculations. Both discretisations have been shown to be related to the 2-component KP hierarchy \cite{Sadakane}. However, as we shall see, there exists a remarkably simple discretisation which, although originally discovered in the context of the 2-component KP hierarchy \cite{DateIV}, can be obtained directly from the HM equation \eqref{HMs} through a new reduction technique. From a numerical point of view, this discretisation has the added benefit of possessing a semi-continuous limit which yields a time-discretisation of the NLS equation, as opposed to the Ablowitz-Ladik equation (or related discretisations) which only offer semi-discrete equations that are space-discretisations of NLS. Unfortunately, apart from a single oblique reference to it at the end of \cite{Suris}, this discrete form of NLS seems to have been completely forgotten over the years.

The mathematical object that is of crucial importance in the reduction technique we shall introduce in this paper, is the so-called {\em squared eigenfunction potential}.
\begin{theorem}\label{Om-def}
For any tau function $\tau$ for the discrete KP hierarchy, given an eigenfunction $\psi$ and an adjoint eigenfunction $\psi^*$ associated to $\tau$, there exists a squared eigenfunction potential  $\Omega(\psi,\psi^*)$, uniquely defined by the difference relations
\begin{equation*}
\Delta_{\ell_j}\Omega(\psi,\psi^*) =~\! \psi_{\ell_j} \psi^*\qquad (j= 1, 2, \hdots)~\!,
\end{equation*}
up to an additive constant.
\end{theorem}
Here, we denote by $\Delta_{\ell_j}$ the forward difference operator in the $\ell_j$ direction : 
\begin{equation}
\Delta_{\ell_j} f(\ell_j) = a_j \big[ f(\ell_j+1) - f(\ell_j) \big]~\!.\label{Del}
\end{equation}
The well-definedness of the squared eigenfunction potential for arbitrary tau functions is shown in \cite{RW-IP, RW-JMP1} using the fermionic representation of the KP hierarchy and properties of the KP vertex operators. The connection between $\Omega(\psi,\psi^*)$ and the KP vertex operators was first noted in \cite{Adlervmb}. Whereas the proof of the existence of the squared eigenfunction potential requires relatively sophisticated techniques, the fact that the defining equations of the potential are compatible, can be checked by elementary calculations. For example, using the evolution along the lattice directions $m$ and $n$ in the linear problem \eqref{dKP-LP} and its adjoint \eqref{dKP-aLP}, one immediately finds that
\begin{align*}
\Delta_n (\psi_m \psi^*) = \frac{\mu \nu}{\mu-\nu}~\!\frac{\tau_m \tau_n}{\tau \tau_{mn}}~\!\left[ \psi_n \psi^*_n - \psi_m \psi^*_m \right]~\!,
\end{align*}
which is symmetric in the $m$ and $n$ directions. Hence, $\Delta_n (\psi_m \psi^*) = \Delta_m (\psi_n \psi^*)$.

The connection with the KP vertex operator becomes apparent in the following theorem:
\begin{theorem}\label{BDT}
Given a tau function $\tau$ for the discrete KP hierarchy and an associated squared eigenfunction potential $\Omega(\psi,\psi^*)$, the map
$$ \tau ~\!\mapsto~\!  \tau \times \Omega(\psi,\psi^*)~\!,$$
is a map between tau functions, or in other words, $\hat\tau= \tau~\! \Omega(\psi,\psi^*)$ satisfies the discrete KP hierarchy.
\end{theorem}
The proofs given in \cite{RW-JMP1, RW-IP} also deal with the fermionic representation of this map and with its interpretation as an extension of the action of a solitonic vertex operator. However, it turns out that a simple, algebraic, verification that $\hat\tau$ indeed satisfies the HM equation \eqref{HMs} is possible. It suffices to verify that: 
\begin{align*}
(\mu - \nu)~\! \hat\tau_{l}& \hat\tau_{mn}~\! + ~\! (\nu -\lambda)~\! \hat\tau_m \hat\tau_{ln} ~\!  +~\! (\lambda - \mu)~\! \hat\tau_n \hat\tau_{lm}\\
&=~\! \frac{\tau_l \tau_m \tau_n}{\tau~\! \psi^*} \left[ \Omega_l \big( \mu \psi^*_m - \nu \psi^*_n\big) \Omega_{mn} + \text{cyclic~ perm.} \begin{pmatrix}l\rightarrow m\rightarrow n\rightarrow l\\ \lambda\rightarrow\mu\rightarrow\nu\rightarrow\lambda\end{pmatrix} \right]\\
&\qquad=~\!\frac{\tau_l \tau_m \tau_n}{\tau~\! \psi^*} \left[ \mu \Omega_l \psi^*_m \Omega_n - \nu \Omega_m \psi^*_n \Omega_l + \text{cyclic~ perm.}  \right] \equiv 0~\!,
\end{align*}
because of the HM equation \eqref{HMs} and the linear equations \eqref{dKP-aLP}. This result was first obtained in \cite{Nimmo-DT}, using the notion of a binary Darboux transformation.

Under the Miwa-transformation \eqref{Miwatransf}, the continuum limit of the defining relations for the squared eigenfunction potential yields the coupled equations \cite{RW-IP}
\begin{gather*}
\partial_{x_j}\big(\Omega(\psi,\psi^*)\big) ~\!=~\! j \psi^* p_{j-1}(\tilde\partial)(\psi) - \sum_{k=1}^{j-1} \partial_{x_k}[\psi^* p_{j-k-1}(\tilde\partial) (\psi) ]\quad(j=1,2,\hdots)~\!,
\end{gather*}
subject to the conditions
\begin{equation}
p_j(-\tilde\partial)(\psi)~\!=~\! \psi~\! p_{j-1}(-\tilde\partial) (\partial_{x_1}\log\tau)\quad(j=2,3,\hdots)\label{KP-LP}~\!,
\end{equation}
and
\begin{equation}
p_j(\tilde\partial)(\psi^*)~\!=~\! -\psi^*~\! p_{j-1}(\tilde\partial) (\partial_{x_1}\log\tau)\quad(j=2,3,\hdots)\label{KP-aLP}~\!,
\end{equation}
where $p_j(\pm\tilde\partial)$ denotes the weight $j$ Schur polynomial, obtained from the generating function $e^{\sum_{i=1}^{+\infty} u_i\lambda^i} = \sum_{j=0}^{+\infty} p_j(u_1,u_2, u_3, ...) \lambda^j$, in the (weighted) partial differential operators $\pm \partial_{x_1}, \pm\partial_{x_2}/2,  \pm\partial_{x_3}/3, \cdots$. The systems \eqref{KP-LP} and \eqref{KP-aLP} are very compact forms of the Zakharov-Shabat system for the KP hierarchy and its adjoint. It can be shown that the above relations in fact define a (formal) exact differential
\begin{equation*}
d\Omega(\psi,\psi^*) ~\!=~\! \psi\psi^* dx + (\psi_{x_1}\psi^*  - \psi \psi_{x_1}^*) dx_2 + \cdots~\!,
 \end{equation*}
(where subscripts stand for partial derivatives), hence the name squared eigenfunction potential. It is well-known \cite{KonopelchenkoS} that the $x_1$ derivative of the squared eigenfunction $\psi\psi^*$ is the generator for most of the symmetries of the KP hierarchy, a fact that is of course rooted in its deep relation to the KP vertex operator, but which can also be understood in the following way. Given a KP tau function $\tau$, the (continuous) KP equation is usually expressed in terms of the variable $u= \partial_{x_1}^2 \log \tau$. Now, since Theorem \ref{BDT} also holds at the continuum limit, one can use the tau function $\hat\tau$ to obtain a new solution to the KP equation: $\hat u= \partial_{x_1}^2 \log \hat\tau$. As the squared eigenfunction potential $\Omega$ is only defined up to an arbitrary constant, it follows that the expansion
$$\hat u ~\!=~\! \partial_{x_1}^2 \log (\tau \Omega + \tau/\eps)~\!=~\! u + \partial_{x_1}^2 \log(1+\eps~\!\Omega) = u + \eps~\! (\psi\psi^*)_{x_1} +{\mathcal O}(\eps^2)~\!,$$
satisfies the KP equation to all orders in $\eps$, and thus also up to first order. Hence, $(\psi\psi^*)_{x_1}$ generates a generalized symmetry for the KP equation.

Standard reductions of the KP equation (hierarchy) to 1+1 dimensional integrable equations rely on a simple invariance with respect to translations in a particular direction. E.g., requiring that $u_{x_2}=0$ (or, more fundamentally that $\tau_{x_2}\sim\tau$) reduces KP to the KdV equation. A {\em symmetry reduction} of the KP equation (and of the hierarchy as a whole) relies on constraining the general symmetry generator $(\psi\psi^*)_{x_1}$ to translations in a particular direction: $u_{x_k} = (\psi\psi^*)_{x_1}$. Such a restriction involving both the solutions to the linear equations $(\psi, \psi^*)$ and those of the nonlinear equations $(u)$ is called a {\em symmetry constraint} \cite{KonopelchenkoS} or a $k$-constraint \cite{Cheng}.

\section{Main results}\label{mainresults}
In \cite{Loris1}, symmetry constraints were reinterpreted as constraints on the KP tau functions of the form: $\tau_{x_k} = \tau~\!\Omega(\psi,\psi^*)$. The aim of this paper is to implement a similar constraint in the discrete case, and to obtain explicit forms for the systems that result from the first few constraints: i.e. for $k=1,2$ and $3$. 
\begin{theorem}\label{maintheorem}
Let $S$ represent an arbitrary shift on the discrete KP lattice, and let $\gamma\in\mathbb{C}^\times$. The constraint
\begin{equation}
\gamma~\!S\tau ~\!=~\! \tau~\!\Omega(\phi,\phi^*)~\!,\label{constraint}
\end{equation}
on a tau function $\tau$ and an eigenfunction $\phi$ and adjoint eigenfunction $\phi^*$ associated to it through the linear equations \eqref{dKP-LP} and \eqref{dKP-aLP}, is compatible with the discrete KP hierarchy.  The constrained hierarchy obtained by imposing this condition on the discrete KP hierarchy and its linear and adjoint linear system, has solutions $\tau, \phi, \phi^*$ that can be expressed by bi-directional Casorati determinants
\begin{equation}
\tau^{(N)}~\!:=~\! \Big|\Delta_S^{i-1} \Delta^{j-1} f(\ell) \Big|_{i,j=1..N}~\!,\qquad \tau^{(0)} := 1~\!,\label{bidircaso}
\end{equation}
where $\Delta$ is the (forward) difference operator \eqref{Del} in an arbitrary but fixed direction on the lattice and $\Delta_S:= \gamma (S - 1)$. The function $f(\ell)$ should satisfy the dispersion relations $\Delta_{\ell_j} f(\ell) = \Delta_{\ell_i} f(\ell)$ for all possible directions $\ell_i,\ell_j$ on the lattice. 
Solutions $\tau, \phi, \phi^*$ to the constraint are given by:
$$\tau = \tau^{(N)}~\!,\quad \phi=\frac{S\tau^{(N-1)}}{\tau^{(N)}}~\!,\quad \phi^*=\frac{\tau^{(N+1)}}{\tau^{(N)}}~\!,$$
for an arbitrary positive integer $N$.
\end{theorem}
These solutions can in fact be extended to the case $N=0$ by assuming $\tau^{(-1)}=0$ whenever this makes sense in the actual equations.

\smallskip
The simplest example of such a constraint leads to the discretisation of the NLS equation mentioned in the previous section.
\begin{prop}\label{dNLSprop}
The case $\gamma=\nu, S = T_n$,  in terms of the shift operator $T_n$ in the $n$ direction, yields the following integrable discretisation of the NLS equation:
\begin{equation}
\left\{\begin{matrix}
\mu \phi_m^* - \nu \phi_n^* = \dfrac{(\mu-\nu) \phi^*}{1-\frac{1}{\mu\nu} \phi^*\fie_{mn}}\\[-1.5mm]\\
\mu \fie_n - \nu \fie_m = \dfrac{(\mu-\nu) \fie_{mn}}{1-\frac{1}{\mu\nu} \phi^*\fie_{mn}}~\!.
\end{matrix}\right.\label{dNLS}
\end{equation}
This equation can be cast into the Hirota bilinear form 
\begin{gather}
\mu\nu ( \tau \tau_{mn} - \tau_m \tau_n ) = \tau^* \tau_{mn}' \label{dNLS-BF1}\\
(\mu-\nu) \tau_{mn} \tau^* = \mu \tau_n \tau_m^* - \nu \tau_m \tau_n^*\label{dNLS-BF2}\\
(\mu-\nu) \tau \tau_{mn}' = \mu \tau_m \tau_n' - \nu \tau_n \tau_m'\label{dNLS-BF3}
\end{gather}
by means of the transformation
\begin{equation*}
\phi^* = \frac{\tau^*}{\tau}~\!,\quad \fie =\frac{\tau'}{\tau}~\!.
\end{equation*}
\end{prop}
This discrete system was discovered 30 years ago \cite{DateIV} but has almost been completely forgotten since. It was recently rediscovered in \cite{Hattori}.

A first continuum limit, based on \eqref{contlim}, where $|\mu|\rightarrow\infty$, yields the differential-difference system \cite{Hattori}
\begin{equation}
\left\{\begin{array}{l}
\nu(\phi_n^* - \phi^*) = \phi_{x_1}^* - \frac{1}{\nu} (\phi^*)^2 \fie_n\\[-2mm]\\
\nu(\fie_n - \fie) = (\fie_n)_{x_1} + \frac{1}{\nu} \phi^* (\fie_n)^2
\end{array}\right.\label{sdNLS}
\end{equation}
which is a time-discretisation of the NLS equation. The NLS equation in its so-called complexified form, is obtained at the next continuum limit ($|\nu|\rightarrow\infty$)
\begin{equation}
\left\{
\begin{array}{l}
\phi_{x_2}^* = -( \phi_{x_1x_1}^* + 2 (\phi^*)^2\fie )\cr
\fie_{x_2} = \fie_{x_1x_1} +2 \phi^* \fie^2\quad~\!.
\end{array}
\right.\label{NLS}
\end{equation}
It will be shown in section \ref{extensions}, that the discrete system \eqref{dNLS} can be used to obtain a compact form of the entire continuous NLS hierarchy. 

The continuum limits of the tau functions \eqref{bidircaso} in this case are of course nothing but the well-known bi-directional Wronski determinants 
\begin{equation}
\tau = \tau^{(N)}~\!:=~\! \Big| \big(\dfrac{\partial}{\partial_{x_1}}\big)^{i+j-2} f \Big|_{i,j=1..N}~\!,\quad \tau'=\tau^{(N-1)}~\!,\quad \tau^*=\tau^{(N+1)}~\!,\label{bidirWronski}
\end{equation}
that solve the NLS equation in its bilinear form:
\begin{equation*}
\left\{\begin{array}{l}
\dfrac{1}{2} D_{x_1}^2 \tau\cdot\tau = \tau' \tau^*\\[-2mm]\\
\big( D_{x_2} - D_{x_1}^2 \big) \tau'\cdot\tau = 0\\[-3mm]\\
\big( D_{x_2} + D_{x_1}^2 \big) \tau^*\cdot\tau = 0~\!,
\end{array}\right.
\end{equation*}
which is obtained from (\ref{dNLS-BF1}--\ref{dNLS-BF3}), at the continuum limit. The function $f(x_1, x_2, \hdots)$ that defines the determinant \eqref{bidirWronski}, is required to satisfy the dispersion relations $f_{x_k} = (-1)^{k+1} f_{k x_1}$ (for arbitrary $k=1, 2, \hdots$).

Moreover, the tau functions that satisfy the bilinear system \eqref{dNLS-BF1}--\eqref{dNLS-BF3} also yield solutions to a discretisation of the Broer-Kaup equation \cite{BK}:
\begin{prop}\label{dBKprop}
The system
\begin{equation}
\left\{\begin{array}{l}
H_{mn} = \dfrac{H_n U_n (\mu - \nu H_m)}{U_m (\mu - \nu H_n)}\\[-1.5mm]\\
U=\dfrac{\mu\nu(H_n - H_m)}{(\mu - \nu H_n)(\mu - \nu H_m)}+\dfrac{\mu U_m}{(\mu - \nu H_m)}-\dfrac{\nu U_n H_n}{(\mu - \nu H_n)}
\end{array}\right.\label{dBK}
\end{equation}
is an integrable discretisation of the Broer-Kaup equation, 
\begin{equation}
\left\{\begin{array}{l}
h_{x_2} = \big(h_{x_1} + 2 u + h^2 \big)_{x_1}\\[-3mm]\\
u_{x_2} = \big(2 u h - u_{x_1}\big)_{x_1}~\!.\label{BK}
\end{array}\right.
\end{equation}
It is related to the tau functions $\tau$ and $\tau'$ of the dNLS equation \eqref{dNLS} through the dependent variable transformations:
\begin{equation*}
U= \frac{\tau~\! \tau_{mn}}{\tau_m\tau_n}~\!,\quad H= \dfrac{\tau_n\tau_m'}{\tau_m\tau_n'}~\!,
\end{equation*}
and $u=(\log \tau)_{x_1 x_1}$ and $h = (\log \tau'/\tau)_{x_1}$ for the continuous system. The discrete Broer-Kaup system \eqref{dBK}  corresponds to the Hirota bilinear form:
\begin{gather*}
(\mu-\nu) \tau \tau_{mn}' = \mu \tau_m \tau_n' - \nu \tau_n \tau_m'\label{dBK-BF1}\\[-4mm]\\
(\mu-\nu) \tau \tau_{mnn}' = \mu \tau_m \tau_{nn}' - \nu \tau_n \tau_{mn}'\label{dBK-BF2}~\!,
\end{gather*}
which has solutions $\tau= \tau^{(N)}$ and $\tau'= \tau^{(N-1)}$ with $\tau^{(N)}$ as in \eqref{bidircaso}.
\end{prop}

An integrable discretisation of the Yajima-Oikawa system \cite{YO} is obtained from a second order shift operator $S$:
\begin{prop}\label{dYOprop}
The case $\gamma=-\mu\nu, S = T_mT_n$, in terms of shift operators $T_m$ and $T_n$ in the $m$ and $n$ directions, yields an integrable discretisation of the Yajima-Oikawa system if one imposes the restriction $\nu=-\mu$ on the lattice parameters:
\begin{equation}
\left\{\begin{array}{l}
2 \mu^3 (U_{m'n'} - U) = \fie_m \phi_{n'}^*-\fie_n\phi_{m'}^*\\[-2mm]\\
\phi_m^* + \phi_n^* = 2 U \phi^*\\[-1.5mm]\\
\fie_m+\fie_n = 2 U \fie_{mn}
\end{array}\right.~\!.\label{dYO}
\end{equation}
This equation can be cast into the Hirota bilinear form 
\begin{gather}
\mu^3( \tau \tau_{mmn} - \tau_m \tau_{mn} ) = \tau^* \tau_{mmn}' \label{dYO-BF1}\\
2 \tau_{mn} \tau^* = \tau_n \tau_m^* + \tau_m \tau_n^*\label{dYO-BF2}\\
2 \tau \tau_{mn}' = \tau_m \tau_n' + \tau_n \tau_m'\label{dYO-BF3}
\end{gather}
by means of the transformation
\begin{equation*}
U= \frac{\tau~\! \tau_{mn}}{\tau_m\tau_n}~\!,\quad\phi^* = \frac{\tau^*}{\tau}~\!,\quad \fie =\frac{\tau'}{\tau}~\!.
\end{equation*}
\end{prop}
At the continuum limit, system \eqref{dYO} indeed yields the Yajima-Oikawa system:
\begin{equation}
\left\{\begin{array}{l}
u_{x_2} = (\phi^*\fie)_{x_1}\\[-2mm]\\
\phi_{x_2}^* = -( \phi_{x_1x_1}^* + 2 u \phi^* )\\[-2mm]\\
\fie_{x_2} = \fie_{x_1x_1} +2 u \fie
\end{array}\right.\label{YO}
\end{equation}
where $u=(\log \tau)_{x_1 x_1}$. In this case, the continuum limits of the tau functions \eqref{bidircaso} are of course again bi-directional Wronski determinants, 
\begin{equation*}
\tau = \tau^{(N)}~\!:=~\! \Big| \dfrac{\partial^{i+j-2}}{\partial{{x_2}^{i-1}}\partial{{x_1}^{j-1}}} f(\ell) \Big|_{i,j=1..N}~\!,\quad \tau'=\tau^{(N-1)}~\!,\quad \tau^*=\tau^{(N+1)}~\!,
\end{equation*}
for $f(x_1, x_2, \hdots)$ subject to the dispersion relations $f_{x_k} = (-1)^{k+1} f_{k x_1}$ ($k=1, 2, \hdots$). These tau functions solve the system of Hirota bilinear equations
\begin{equation*}
\left\{\begin{array}{l}
\dfrac{1}{2} D_{x_1}D_{x_2} \tau\cdot\tau = \tau' \tau^*\\[-2mm]\\
\big( D_{x_2} - D_{x_1}^2 \big) \tau'\cdot\tau = 0\\[-3mm]\\
\big( D_{x_2} + D_{x_1}^2 \big) \tau^*\cdot\tau = 0~\!,
\end{array}\right.
\end{equation*}
which is obtained as the continuum limit of (\ref{dYO-BF1}--\ref{dYO-BF3}).

\section{Solutions, Lax pairs and continuum limits}\label{proofs}
In this section we shall prove the theorem and propositions presented in section \ref{mainresults}. Let $\phi$ be an eigenfunction for a given discrete KP tau function $\tau$. In other words, $\phi$ satisfies the relation
\begin{equation}
\phi_{lk} = \frac{1}{a_l-a_k}\frac{\tau_l\tau_k}{\tau\tau_{lk}} \big(a_l\phi_k - a_k\phi_l\big)\label{phidefeq}~\!,
\end{equation}
for all possible choices of two (different) directions $l$ and $k$ on the discrete KP lattice. Similarly, let $\phi*$ be an adjoint eigenfunction for that same tau function:
\begin{equation}
\phi^* = \frac{1}{a_l-a_k}\frac{\tau_l\tau_k}{\tau\tau_{lk}} \big(a_l\phi_l^* - a_k\phi_k^*\big)\label{phistdefeq}~\!,
\end{equation}
for all choices of different lattice directions $l$ and $k$. These functions will be parametrized in terms of $\tau$ and new functions $\bar\tau$ for ($\phi$) and $\tau^*$ for ($\phi^*$) :
\begin{equation}
\phi:=\frac{\bar\tau}{\tau}~\!,\quad\phi^*:= \frac{\tau^*}{\tau}~\!.\label{phietcdef}
\end{equation}
The linear and adjoint linear equations \eqref{phidefeq} and \eqref{phistdefeq} then give rise to Hirota bilinear relations for all different pairs of lattice directions for these new functions. In particular, in the $m$ and $n$ directions we find
\begin{equation}
\label{taubdefeq}
(\mu-\nu) \tau {\bar\tau}_{mn}  = \mu \tau_m{\bar\tau}_n - \nu \tau_n {\bar\tau}_m~\!,
\end{equation}
for $\tau$ and $\bar\tau$, and
\begin{equation}
\label{taustdefeq}
(\mu-\nu) \tau_{mn} \tau^* = \mu \tau_n\tau_m^* - \nu \tau_m \tau_n^*~\!,
\end{equation}
for $\tau$ and $\tau^*$.

First, we shall prove the following crucial Lemma:
\begin{lemma}\label{crucial lemma}
If $\tau$, $\phi$ and $\phi^*$ satisfy the constraint
$$\gamma~\!S\tau ~\!=~\! \tau~\!\Omega(\phi,\phi^*)~\!,$$
for some shift $S$ on the discrete KP lattice, and for some $\gamma\in\mathbb{C}^\times$, then 
\begin{equation}
\fie := \dfrac{\tau'}{\tau}~\!,\label{fiedef}
\end{equation}
where $\tau' := S^{-1} \bar\tau$, also satisfies the defining equation \eqref{phidefeq} for an eigenfunction.
\end{lemma}
\begin{proof}
We give the proof for the $m, n$ directions, which is representative for all possible pairs of directions.
On account of the definition of the potential $\Omega$ in Theorem (\ref{Om-def}), taking (forward) differences of the constraint in the $m$ and $n$ directions, one obtains
\begin{gather*}
\mu \gamma~\! \left( \dfrac{S\tau_m}{\tau_m} - \dfrac{S\tau}{\tau} \right)~\!=~\! \phi^* \phi_m ~\!=~\! \dfrac{\tau^* (S\tau_m')}{\tau \tau_m}\\
\nu \gamma~\! \left( \dfrac{S\tau_n}{\tau_n} - \dfrac{S\tau}{\tau} \right)~\!=~\! \phi^* \phi_n ~\!=~\! \dfrac{\tau^* (S\tau_n')}{\tau \tau_n}~\!,
\end{gather*}
from which it follows that
\begin{equation*}
\tau \big[ \mu S(\tau_m\tau_n') - \nu S(\tau_n\tau_m') \big] = (S\tau) \big[ \mu \tau_m \bar\tau_n - \nu \tau_n \bar\tau_m \big]~\!.
\end{equation*}
Because of \eqref{taubdefeq}, i.e. because of the fact that $\phi$ is a discrete KP eigenfunction, the r.h.s. in this last expression is nothing but $(\mu-\nu) \tau\bar\tau_{mn}$ or $(\mu-\nu) \tau (S\tau_{mn}')$, and one finds that
\begin{equation}
\label{taupdefeq}
(\mu-\nu) \tau \tau_{mn}'  = \mu \tau_m \tau_n' - \nu \tau_n  \tau_m'~\!.
\end{equation}
Hence, $\fie$ -- as defined by \eqref{fiedef} -- satisfies
\begin{equation}
\fie_{mn} = \frac{1}{\mu-\nu}\frac{\tau_m\tau_n}{\tau\tau_{mn}} \big(\mu\fie_n - \nu\fie_m\big)~\!.\label{fiedefeq}
\end{equation}
\end{proof}

\begin{remark}\label{fie-phi}
Note that this proof can also be read backwards: i.e., under the constraint, one also has that the fact that $\fie = \tau'/\tau$ satisfies the discrete KP linear system for the tau function $\tau$, necessarily implies that $\phi = (S\tau')/\tau$ satisfies these equations as well.
\end{remark}
\begin{remark}
In fact, a similar property holds for the function $\tau^*$, which can be seen to satisfy: 
\begin{equation*}
(\mu-\nu) (S\tau_{mn}) \tau^* = \mu (S\tau_n)\tau_m^* - \nu (S\tau_m) \tau_n^*~\!.
\end{equation*}
Hence, one finds that the ratio $\tau^*/(S\tau)$ satisfies the adjoint linear equations \eqref{phistdefeq}, not for the tau function $\tau$ itself, but for the shifted tau function $S\tau$. In fact, this observation amounts to an explicit verification of the commutativity of the following Bianchi diagram for (adjoint) Darboux transformations \cite{RW-IP}

\resizebox{.225cm}{!}{\mbox{
\begin{picture}(1,2)(-100,140)
\put(0,100){\makebox(85,-1)[t]{$\bigF{\bar\tau}$}}
\put(50,105){\vector(2,1){50}}
\put(89,140){\makebox(-43,-6)[t]{$\bigF{1/S\fie}$}}
\put(50,85){\vector(2,-1){50}}
\put(89,50){\makebox(-16,20)[r]{$\bigF{1/\phi}$}}
\put(135,145){\makebox(-46,-23)[]{$\bigF{\widehat\tau}$}}
\put(135,45){\makebox(-46,23)[]{$\bigF{\tau}$}}
\put(122,130){\vector(2,-1){50}}
\put(174,140){\makebox(-34,-8)[t]{$\bigF{\tau^*/S\tau}$}}
\put(122,60){\vector(2,1){50}}
\put(224,45){\makebox(-146,37){$\bigF{\phi^*}$}}
\put(222,99){\makebox(-80,3)[t]{$\bigF{\tau^*}$}}
\end{picture}
}}\hfill\break\vskip2.4cm\noindent
where $\widehat\tau = \tau~\!\Omega(\phi,\phi^*) \sim S\tau$ under the constraint.
\end{remark}

We now proceed with the proof of Theorem \ref{maintheorem}, for which it suffices to show that the constraint shares (a nontrivial) part of the solution space of the discrete KP hierarchy and its associated linear (and adjoint linear) system.
\begin{proof}[{\bf Proof of Theorem \ref{maintheorem}}]
It is well-known \cite{Ohta-dKP,Nimmo-DT} that the discrete KP hierarchy has tau functions in the form of Casorati determinants
\begin{equation*}
\tau^{(N)} ~\!=~\! \Big| \Delta^{j-1} f^{(i)}(\ell) \Big|_{i,j=1..N}~\!,
\end{equation*}
defined in terms of functions $f^{(i)}(\ell)$ that are required to satisfy the dispersion relations $\Delta_{\ell_j} f^{(i)}(\ell) = \Delta_{\ell_k} f^{(i)}(\ell)$ for all possible directions $\ell_j, \ell_k$ on the lattice. The difference operator $\Delta$ acts in a fixed, but arbitrary lattice direction, and the value of the resulting determinant is obviously independent of the particular choice of lattice direction. Furthermore, the ratios ${\tau^{(N-1)}}/{\tau^{(N)}}$ and ${\tau^{(N+1)}}/{\tau^{(N)}}$ are known to satisfy, respectively, the linear equations \eqref{phidefeq} and adjoint linear equations \eqref{phistdefeq} associated to $\tau=\tau^{(N)}$. In fact, the functions $f^{(i)}(\ell)$ can be thought of as solutions to the adjoint linear system \eqref{phistdefeq}, for a vacuum tau function $\tau(\ell)=1$, from which the Casorati determinants are then constructed through iterated (adjoint) Darboux transformations \cite{Nimmo-DT,RW-JMP1}. 

It now suffices to find a restriction on the functions $f^{(i)}(\ell)$, besides the dispersion relations, such that the resulting Casorati determinants can simultaneously satisfy the discrete KP hierarchy and the constraint. In fact, Lemma \ref{crucial lemma} offers an important clue. In the simplest non-trivial case, i.e. $N=2$, one has that the pair of tau functions 
\begin{equation*}
(\tau, \tau') ~\!=~\! \big( \begin{vmatrix} f^{(1)} & \Delta f^{(1)}\\  f^{(2)} & \Delta f^{(2)}\end{vmatrix} , f^{(1)}\big)~\!,
\end{equation*} 
defines a solution to the linear system \eqref{phidefeq} for the discrete KP hierarchy. However, as pointed out in Remark \ref{fie-phi}, the constraint implies that the pair 
\begin{equation*}
(\tau, S\tau') ~\!=~\! \big( \begin{vmatrix} f^{(1)} & \Delta f^{(1)}\\  f^{(2)} & \Delta f^{(2)}\end{vmatrix} , Sf^{(1)}\big)~\!,
\end{equation*} 
should also define an eigenfunction. The only conceivable way in which this can happen in a generic fashion, without any conflict with the structure of the solutions for the discrete KP hierarchy, is when $f^{(2)}$ is in fact nothing but  $Sf^{(1)}$, which is the obvious choice that restores the symmetry of this otherwise rather peculiar situation. Hence the ansatz :
$$\tau = \tau^{(N)}~\!,\quad \fie=\frac{\tau^{(N-1)}}{\tau^{(N)}}~\!,\quad \phi^*=\frac{\tau^{(N+1)}}{\tau^{(N)}}~\!,$$
for the bi-directional Casorati determinants
\begin{equation*}
\tau^{(N)}~\!:=~\! \Big|\Delta_S^{i-1} \Delta^{j-1} f(\ell) \Big|_{i,j=1..N}~\!.
\end{equation*}
Note that, if the constraint is indeed verified for such Casorati determinants, Remark \ref{fie-phi} implies that $\phi = (S\tau^{(N-1)})/\tau^{(N)}$ is an eigenfunction for $\tau^{(N)}$. 

All that remains then  is the verification that the constraint is indeed satisfied. Because of the definition of the potential $\Omega$ (Theorem \ref{Om-def}), it is sufficient to verify the constraint in its difference form, i.e. $\Delta (S\tau/\tau) = \Delta\Omega$, in an arbitrary lattice direction. For example, in the $m$ direction one has that
$$\mu\gamma \big( \tau (S\tau_m) - \tau_m (S\tau) \big) = \tau^* (S\tau_m')$$
should be verified.
Obviously, in the case $N=1$, this is satisfied by $\tau'=1, \tau = f(\ell)$ (with $f(\ell)$ subject to the discrete KP dispersion relations $\Delta_{\ell_j} f(\ell) = \Delta_{\ell_i} f(\ell)$ ) and $ \mu\gamma \big( f~\! (Sf_m) - f_m (S f) \big) = \tau^*$, i.e.
$$\tau^* = \mu \gamma \begin{vmatrix}f & T_m f\\ S f & T_m Sf\end{vmatrix} = \gamma \begin{vmatrix}f & \Delta_m f\\ S f & S \Delta_m f\end{vmatrix} = \begin{vmatrix}f & \Delta_m f\\ \Delta_S f & \Delta_S \Delta_m f\end{vmatrix}~\!,$$
which is nothing but $\tau^{(2)}$.
Expressing the general size $N>1$ determinant $\tau^{(N)}$ in terms of shift operators $S$ and $T_m$ rather than difference operators,
\begin{equation*}
\tau^{(N)}~\!=~\! (\mu\gamma)^{\frac{N (N-1)}{2}}~\Big| S^{i-1} T_m^{j-1} f(\ell) \Big|_{i,j=1..N}~\!,
\end{equation*}
one finds that the condition $\mu\gamma \big( \tau^{(N)} (S\tau_m^{(N)}) - \tau_m^{(N)} (S\tau^{(N)}) \big) = \tau^{(N+1)} (S\tau_m^{(N-1)})$ is nothing but the Jacobi identity for the determinants
\begin{multline*}
(\mu\gamma)^{N (N-1)+1} \left[~\!  \Big| S^{i-1} T_m^{j-1} f(\ell) \Big|_{i,j=1..N}~\Big| S^{i} T_m^{j} f(\ell) \Big|_{i,j=1..N} \right.\\\qquad\qquad\qquad\qquad\qquad\left.- ~\!\Big| S^{i-1} T_m^{j} f(\ell) \Big|_{i,j=1..N}~\Big| S^{i} T_m^{j-1} f(\ell) \Big|_{i,j=1..N}~\!\right]\\
= (\mu\gamma)^{N^2 -N +1}~\! \Big| S^{i}~\! T_m^{j} f(\ell) \Big|_{i,j=1..N-1}~\Big| S^{i-1} T_m^{j-1} f(\ell) \Big|_{i,j=1..N+1} ~\!.
\end{multline*}
Hence, $\tau' = \tau^{(N-1)}$ and $\tau^* =  \tau^{(N+1)}$, which proves Theorem \ref{maintheorem}. 
\end{proof}

Next, we show how to obtain Lax pairs for constrained HM equations. The following Lemma is essential for this task.
\begin{lemma}\label{LPfromHM}
Choose an arbitrary lattice direction $k$ on the KP lattice, with lattice parameter $\kappa$. If we denote down-shifts on the lattice by primed subscripts, then the function
\begin{equation*}
\psi := ~\!\dfrac{\tau_{k'}}{\tau}\!\prod_{~j ~\!(\ell_j\neq k)} \left(\dfrac{a_j-\kappa}{a_j}\right)^{-\ell_j}\label{psi}
\end{equation*}
satisfies the discrete KP linear system \eqref{dKP-LP} for any pair of different directions, distinct from $k$.
Furthermore, 
\begin{itemize}
\item[$i)$] if $\tau$ is a constrained tau function, i.e. if as in Theorem \ref{maintheorem} or Lemma \ref{crucial lemma} it satisfies $\gamma S \tau = \tau \Omega(\phi,\phi^*)$, for some specific eigenfunction $\phi$ and adjoint eigenfunction $\phi^*$, then the eigenfunction $\psi$ will satisfy
\begin{equation}
\kappa \gamma \big[ \psi - \Gamma~\! (S\psi) \big]~\!=~\! (S\fie) \phi_{k'}^* \psi~\!,\label{psiconstraint}
\end{equation}
with $\fie$ as defined in Lemma \ref{crucial lemma}. The constant $\Gamma$ is given by
\begin{equation*}
\Gamma := \frac{S\Big[\prod_{j ~\!(\ell_j\neq k)} \left(\dfrac{a_j-\kappa}{a_j}\right)^{\ell_j}\Big]}{\prod_{j ~\!(\ell_j\neq k)} \left(\dfrac{a_j-\kappa}{a_j}\right)^{\ell_j}}~\!.
\end{equation*}
\item[$ii)$] the quantity 
\begin{equation}
\chi := \frac{1}{\kappa} \phi_{k'}^* \psi\label{chidef}
\end{equation}
that appears in \eqref{psiconstraint} satisfies the eigenfunction potential-like equation
\begin{equation}
\Delta ~\!\chi~\! =~\! \phi^* (T\psi)~\!,\label{chidefeq}
\end{equation}
for a shift  $T$ (and forward difference $\Delta$) in an arbitrary direction, different from the direction $k$.
\end{itemize}
\begin{proof}
Consider the HM equation in the $k, m$ and $n$ directions, downshifted once in all three directions:
\begin{equation*}
(\mu-\nu)~\! \tau_{k'} \tau_{m'n'} + (\nu-\kappa)~\! \tau_{m'} \tau_{k'n'} + (\kappa-\mu) ~\!\tau_{n'}\tau_{k'm'} = 0~\!.
\end{equation*}
Then, setting $\tau_{k'} = \tau ~\!\psi ~\!\prod_{j ~\!(\ell_j\neq k)} \left(\dfrac{a_j-\kappa}{a_j}\right)^{\ell_j}$, one immediately obtains the linear equation 
\begin{equation}
\psi_{mn} ~\!= ~\!\frac{1}{\mu-\nu}~\! \frac{\tau_m \tau_n}{\tau \tau_{mn}}~\!\left[~\!\mu~\!\psi_n - \nu~\!\psi_m~\!\right]~\!,\label{psimn}
\end{equation}
downshifted in $m$ and $n$. As the equations in the other directions can be obtained in a similar fashion, this settles the first part of the Lemma. 

For the second part, take the difference $\Delta_k (S\tau/\tau) = \Delta_k\Omega$ of the constraint, which yields
\begin{equation*}
\kappa \gamma \left( \frac{S\tau}{\tau} -  \frac{S\tau_{k'}}{\tau_{k'}}\right)~\!=~\! \phi \phi_{k'}^*~\!,
\end{equation*}
after a downshift in $k$. Using the above substitution for $\tau_{k'}$ in terms of $\psi$, one obtains
\begin{equation*}
\kappa\gamma \frac{S\tau}{\tau} \big[ 1 - \Gamma \frac{S\psi}{\psi}\big]~\!=~\! \phi \phi_{k'}^* = \frac{S\tau'}{\tau} \phi_{k'}^* = \frac{S\tau}{\tau}~\! (S\fie) \phi_{k'}^*~\!,
\end{equation*}
which proves statement $i)$. As for statement $ii)$, since $\phi^*$ is an adjoint eigenfunction for the discrete KP hierarchy, it satisfies e.g. the equation
$$\phi^* = \frac{1}{\mu-\kappa}~\!\frac{\tau_m\tau_k}{\tau \tau_{mk}} \big[ \mu \phi_m^* - \kappa \phi_k^* \big]~\!,$$
and hence, taking a downshift in the $k$ direction and replacing $\tau_{k'}$ and $\phi_{k'}^*$ in this relation by their respective expressions in terms of $\psi$ and $\chi$, one obtains
$$\chi = \frac{1}{\mu} \big[ \mu \chi_m - \phi^* \psi_m \big]~\!,$$
which proves the Lemma.
\end{proof}
\end{lemma}
Note that the functions $\psi$ and $\chi$ defined in Lemma \ref{LPfromHM}, satisfy a set of equations (\ref{psiconstraint},\ref{chidefeq}) which are -- essentially -- linear in these two fields. They will therefore play a crucial role in the construction of Lax pairs for the constrained systems. 

Let us start with the construction of the Lax pair for the discrete NLS equation \eqref{dNLS}, which can be regarded as proof of its integrability.

\begin{proof}[{\bf Proof of Proposition \ref{dNLSprop}}]
Take the shift operator $S$ and multiplicative constant $\gamma$ in the constraint \eqref{constraint} to be $\gamma=\nu$ and $S = T_n$. In that case, the constraint
$$\nu~\! \tau_n = \tau~\! \Omega(\phi,\phi^*)$$
will yield
$$\tau_{x_1} =  \tau~\! \Omega(\phi,\phi^*)~\!,$$
at the continuum limit. This follows immediately from the fact that the squared eigenfunction potential $\Omega(\phi,\phi^*)$ is only determined up to an (arbitrary) additive constant. The constraint can therefore also be expressed as
$$\Delta_S~\! \tau = \nu (\tau_n - \tau) = \tau~\! \Omega(\phi,\phi^*)~\!,$$
which, under the Miwa-transformation \eqref{contlim}, behaves as
$$\tau_{x_1} + {\mathcal O}(\nu^{-1}) =  \tau~\! \Omega(\phi,\phi^*)~,$$
for $|\nu|\rightarrow\infty$. The $x_1$-derivative of the continuum limit of this constraint is nothing but the symmetry constraint which is used to obtain the NLS equation from the KP hierarchy \cite{KonopelchenkoS} :
$$\dfrac{\partial^2}{\partial x_1^2} \log\tau = \Omega_{x_1} = \phi\phi^*~\!.$$
Note that, at the continuum limit, both $\phi=\tau_n/\tau$ and $\fie=\tau'/\tau$ tend to the same function, at lowest order. We can therefore also write this result in a more suggestive way, using the variable names of the discrete system:
$$\dfrac{\partial^2}{\partial x_1^2} \log\tau = \phi^* \fie~\!.$$
On the contrary, taking e.g. the difference in the $m$ direction of the discrete constraint
$$\Delta_m \big(\frac{\tau_n}{\tau}\big) = \frac{1}{\nu} \phi^*\phi_m~\!,$$
just as in the proof of Lemma \ref{crucial lemma}, we obtain, 
\begin{equation}
\frac{\tau_n \tau_m}{\tau \tau_{mn}}~\!=~\! 1 - \frac{1}{\mu\nu}~\! \phi^* \fie_{mn}~\!,\label{dNLS-diffconst}
\end{equation}
which is the discrete equivalent of the above continuous constraint.

Turning our attention now to the Lax pair, it is readily verified that in this case ($S=T_n$) the constant $\Gamma$ in equation \eqref{psiconstraint} is simply $\Gamma=(\nu-\kappa)/\nu$. Then, introducing the function $\chi$ defined by \eqref{chidef} into that same relation, we obtain a very simple linear equation connecting $\psi$ and $\chi$:
\begin{equation}
\psi_n = \dfrac{1}{\nu-\kappa}~\!\big[ \nu~\! \psi - \fie_n \chi \big]~\!.\label{dNLS-LP-d1}
\end{equation}
Furthermore, the dependence on the $n$ direction of $\chi$ is also known because of \eqref{chidefeq}:
$$\Delta_n \chi = \nu(\chi_n - \chi) = \phi^* \psi_n~\!,$$
which on account of \eqref{dNLS-LP-d1} can be rewritten as
\begin{equation}
\chi_n = \dfrac{1}{\nu-\kappa}~\!\big[ \phi^* \psi + \big(~\! (\nu-\kappa) - \frac{1}{\nu} \phi^* \fie_n\big) \chi \big]~\!.\label{dNLS-LP-d2}
\end{equation}
Let us choose an auxiliary direction (different from $k$) on the lattice, say $m$, and let us calculate the evolution of $\psi$ in that direction from  \eqref{dNLS-LP-d1}. Bearing in mind that both $\psi$ and $\fie$ satisfy the linear equations for the dKP hierarchy (and in particular equations \eqref{psimn} and \eqref{fiedefeq}), and using \eqref{chidefeq} to define the shift of $\chi$ in the $m$ direction, the $T_m$-shift of  \eqref{dNLS-LP-d1} can be re-arranged in the following way:
\begin{gather*}
\psi_{mn} =  \dfrac{1}{\nu-\kappa}~\!\big[ \nu~\! \psi_m - \fie_{mn} \chi_m \big]
\end{gather*}
\begin{multline*}
\Leftrightarrow\quad \frac{1}{\mu-\nu} \frac{\tau_n \tau_m}{\tau \tau_{mn}} \big[ \mu \psi_n-\nu \psi_m \big] ~\!=~\!   \frac{1}{\nu-\kappa} \Big[ \nu (1 - \frac{1}{\mu\nu} \phi^*\fie_{mn}) \psi_m  \\
- \frac{\chi}{\mu-\nu} \frac{\tau_n \tau_m}{\tau \tau_{mn}} \big( \mu \fie_n-\nu \fie_m \big) \Big]~\!,
\end{multline*}
which, using relation \eqref{dNLS-diffconst} for the constraint and equation \eqref{dNLS-LP-d1}, is nothing but
\begin{equation*}
\psi_m = \dfrac{1}{\mu-\kappa}~\!\big[ \mu~\! \psi - \fie_m \chi \big]~\!.\label{dNLS-LP-d3}
\end{equation*}
Note that this is exactly the same relation as that in the $n$ direction, i.e. equation \eqref{dNLS-LP-d1}. Hence, the $m$ evolution of $\chi$ will take the same form as \eqref{dNLS-LP-d2} and we obtain the following system of linear equations for the (vector) function $\Psi:= ~\!^t\ \!\!( \psi~\!~\chi )$:
\begin{gather}
\Psi_m = \frac{1}{\mu-\kappa}~\begin{pmatrix}
\mu & -\fie_m \\
\phi^* & \mu-\kappa - \frac{1}{\mu} \phi^* \fie_m\end{pmatrix} \cdot \Psi~\!,\label{dNLS-LP1}\\[-2mm]\nonumber\\
\Psi_n = \frac{1}{\nu-\kappa}~\begin{pmatrix}
\nu & -\fie_n \\
\phi^* & \nu-\kappa - \frac{1}{\nu} \phi^* \fie_n\end{pmatrix} \cdot \Psi~\!.\label{dNLS-LP2}
\end{gather}
The compatibility condition $T_n \Psi_m = T_m \Psi_n$ of these linear equations yields exactly equation \eqref{dNLS}, independently of $\kappa$. Hence, one can conclude that the linear system (\ref{dNLS-LP1},\ref{dNLS-LP2}) is a Lax pair, with spectral parameter $\kappa$, for the dNLS equation \eqref{dNLS}.

The Hirota bilinear form of the dNLS equation is obtained directly from equations \eqref{dNLS-diffconst}, \eqref{taustdefeq} and \eqref{taupdefeq}.
\end{proof}
Taking the continuum limit \eqref{contlim} for $|\mu|\rightarrow\infty$ of \eqref{dNLS-LP1}, one obtains
\begin{equation}
\Psi_{x_1} = \begin{pmatrix}\kappa & -\fie\\ \phi^* & 0\end{pmatrix}\cdot\Psi\label{NLS-LP1}~\!,
\end{equation}
which together with \eqref{dNLS-LP2} forms a Lax pair for the semi-discrete NLS equation \eqref{sdNLS}. Taking,  finally, the continuum limit $|\nu|\rightarrow\infty$ in \eqref{dNLS-LP2}, subject to \eqref{NLS-LP1}, one obtains
$$\Psi_{x_2} = \begin{pmatrix}\kappa^2+\fie\phi^* & -\fie_{x_1}-\kappa\fie \\ \kappa\phi^*-\phi^*_{x_1} & -\phi^*\fie\end{pmatrix}\cdot\Psi~\!,$$
which together with \eqref{NLS-LP1} constitutes a Lax pair for the NLS equation \eqref{NLS}. Note that in the usual theory of symmetry constraints, the constraint is formulated in terms of a pseudo-differential form of the Lax equations \cite{Cheng}. In fact, as equation \eqref{chidefeq} tells us that the function $\chi$ is very similar to a squared eigenfunction potential defined for the pair $(\psi, \phi^*)$, the first linear equation \eqref{NLS-LP1} can be thought of as saying that $\psi_{x_1} = \kappa\psi - \fie~\! \Omega(\psi,\phi^*)$, which is one way to interpret the relation $\big(\partial_{x_1} + \fie~\! \partial^{-1} \phi^*\big) \psi = \kappa \psi$ in \cite{Cheng}. 

\smallskip
Although the choice of the function $\chi$ \eqref{chidef} turned out to be quite judicious, especially in view of the deeper meaning of the equations \eqref{chidefeq}, one can wonder whether it is the only interesting one. It so happens that there is another interesting candidate for a function that can be used to construct a Lax pair from \eqref{psiconstraint}:
$$\widetilde\chi := \frac{1}{\kappa} \phi_{k'}^* (S\fie) \psi~\!.$$
This choice allows one to completely eliminate the tau function $\tau^*$ from the problem. The proof of Proposition \ref{dBKprop} will illustrate this fact for the case of the reduction to the dNLS equation.
\begin{proof}[{\bf Proof of Proposition \ref{dBKprop}}]
In case $S=T_n$, the change of variables $\chi \mapsto \widetilde\chi$ corresponds to a gauge transformation on the Lax pair (\ref{dNLS-LP1},\ref{dNLS-LP2})
$$\Psi~\mapsto~\widetilde\Psi := \begin{pmatrix}\psi\\\widetilde\chi\end{pmatrix} ~\!=~\! \begin{pmatrix} 1 & 0 \\ 0 & \fie_n\end{pmatrix}\cdot \Psi~\!,$$
which yields the linear equations:
\begin{gather*}
\widetilde\Psi_m = \frac{1}{\mu-\kappa}~\begin{pmatrix}
\mu & -\dfrac{\fie_m}{\fie_n} \\
\phi^*\fie_{mn} & (\mu-\kappa)\dfrac{\fie_{mn}}{\fie_n} - \frac{1}{\mu} \phi^* \dfrac{\fie_m\fie_{mn}}{\fie_n}\end{pmatrix} \cdot \widetilde\Psi~\!,\\[-2mm]\nonumber\\
\widetilde\Psi_n = \frac{1}{\nu-\kappa}~\begin{pmatrix}
\nu & -1 \\
\phi^*\fie_{nn} & (\nu-\kappa)\dfrac{\fie_{nn}}{\fie_n}  - \frac{1}{\nu} \phi^* \fie_{nn}\end{pmatrix} \cdot \widetilde\Psi~\!.
\end{gather*}
If we now introduce the new variables
\begin{gather*}
H := \dfrac{\tau_n \tau_m'}{\tau_m \tau_n'} = \frac{\fie_m}{\fie_n}\label{Hdef}~\!,\\[-3mm]\nonumber\\
U := \dfrac{\tau \tau_{mn}}{\tau_m \tau_n} ~\!, \label{Udef}
\end{gather*}
which, because of the constraint \eqref{dNLS-diffconst} can also be written as
\begin{equation*}
U = \dfrac{1}{1-\frac{1}{\mu\nu} \phi^*\fie_{mn}}\label{Ualtdef}~\!,
\end{equation*}
one readily obtains the linear system:
\begin{gather}
\widetilde\Psi_m = \frac{1}{\mu-\kappa}~\begin{pmatrix}
\mu & -H \\[-2mm] \\
\mu\nu\dfrac{U-1}{U} & \dfrac{\mu-\kappa}{\mu-\nu}~\!\dfrac{(\mu-\nu H)}{U} +\nu \dfrac{H (1-U)}{U}\end{pmatrix} \cdot \widetilde\Psi~\!,\label{dBK-LP1}\\[-2mm]\nonumber\\
\widetilde\Psi_n = \frac{1}{\nu-\kappa}~\begin{pmatrix}
\nu & -1 \\[-2mm] \\
\mu\nu\dfrac{U-1}{U H_n} & \dfrac{\nu-\kappa}{\mu-\nu}~\!\dfrac{(\mu-\nu H)}{U H_n} +\nu \dfrac{(1-U)}{U H_n}\end{pmatrix} \cdot \widetilde\Psi~\!.\label{dBK-LP2}
\end{gather}

The compatibility condition of this linear system is independent of the spectral parameter $\kappa$ and is identical to the discrete Broer-Kaup equation \eqref{dBK}. The tau functions $\tau$ and $\tau'$ that define the solutions $U$ and $H$ are of course those of the dNLS equation and the Hirota bilinear equations they satisfy are \eqref{taupdefeq} and \eqref{taubdefeq} with $\bar\tau = S\tau' = \tau_n'$.

The continuum limit is obtained through the ansatz
$$ U = 1 + \dfrac{1}{\mu\nu}~\! u~\!, \quad H = 1+ \dfrac{\nu-\mu}{\mu\nu} ~\!h~\!,$$
which yields the tau function expressions
$$u= (\log\tau)_{x_1 x_1}~\!,\quad h=\big(\log\dfrac{\tau'}{\tau}\big)_{x_1}~\!,$$
for the solutions $u$ and $h$ to the (continuous) Broer-Kaup system \eqref{BK}.
\end{proof}
Note that a Lax pair for the BK equation \eqref{BK} can be easily obtained at the continuum limit $|\mu|,|\nu| \rightarrow \infty$ of the linear equations \eqref{dBK-LP1} and\eqref{dBK-LP2}:
\begin{gather}
\widetilde\Psi_{x_1} = \begin{pmatrix}\kappa & -1 \\ u & h\end{pmatrix}\cdot\widetilde\Psi\label{BK-LP1}~\!,\\
\widetilde\Psi_{x_2} = \begin{pmatrix}\kappa^2 + u & -\kappa-h \\ \kappa u + u h - u_{x_1} & -u -h_{x_1} - h^2\end{pmatrix}\cdot\widetilde\Psi\label{BK-LP2}~\!.
\end{gather}

This construction of a discretisation of the Broer-Kaup equation through a gauge transformation and subsequent change of variables applied to the Lax pair for the (discrete) NLS equations, in fact offers a remarkably faithful analogy to the continuous situation. As described in \cite{JaulentM, Sachs}, the Lax pair for the Broer-Kaup equation (\ref{BK-LP1},\ref{BK-LP2}) can be obtained from that for the NLS equation in the context of a rather general scheme of gauge transformations for so-called energy dependent scattering problems.

\smallskip
Finally, let us consider the case of the discrete Yajima-Oikawa system, for which $S=T_m T_n$:
\begin{proof}[{\bf Proof of proposition \ref{dYOprop}}]
In this case, once again exploiting the freedom in the squared eigenfunction potential, the constraint
$$\gamma~\! (\tau_{mn} - \tau) = \tau~\! \Omega(\phi,\phi*)$$
can, under the Miwa-transformation \eqref{contlim}, be expanded for large $\mu$ and $\nu$  as
$$\gamma~\! \big( (\dfrac{1}{\mu}+\dfrac{1}{\nu})~\! \tau_{x_1} + \frac{1}{2}(\dfrac{1}{\mu^2}+\dfrac{1}{\nu^2})~\! \tau_{x_2} + \frac{1}{2}(\dfrac{1}{\mu}+\dfrac{1}{\nu})^2~\! \tau_{x_1x_1} + \cdots\big) =  \tau~\! \Omega(\phi,\phi*)~,$$
which clearly cannot yield the required constraint on the $x_2$ derivative of $\tau$ unless $\nu=-\mu$. Hence, imposing $\nu=-\mu$, we choose $\gamma=-\mu\nu=\mu^2$, so that we have
$$\tau_{x_2} + {\mathcal O}(\mu^{-1})  =  \tau~\! \Omega(\phi,\phi^*)~,$$
and therefore
$$\tau_{x_2} =  \tau~\! \Omega(\phi,\phi*)~\!,$$
or 
$$\dfrac{\partial^2}{\partial x_1\partial x_2} \log\tau = \Omega_{x_1} = \phi\phi^* = \phi^*\fie~\!,$$
at the continuum limit. 

Taking the difference in the $m$ direction of the discrete constraint, one easily obtains
$$\mu^3( \tau \tau_{mmn} - \tau_m \tau_{mn} ) = \tau^* \tau_{mmn}' ~\!,$$
which is the discrete equivalent of the above continuous constraint. It is also the first equation in the Hirota bilinear form (\ref{dYO-BF1}--\ref{dYO-BF3}) presented in Proposition \ref{dYOprop}. The remaining two bilinear equations \eqref{dYO-BF2} and \eqref{dYO-BF3} are nothing but the equations \eqref{taustdefeq} and \eqref{taupdefeq} for $\nu=-\mu$.

Obtaining the Lax pair for the discrete Yajima-Oikawa equation turns out to be a lot easier than in the case of the discrete NLS equation. First of all, for $S=T_mT_n$ and $\nu=-\mu$, the constant $\Gamma$ in Lemma \ref{LPfromHM} takes the value
$$\Gamma = \dfrac{\mu^2-\kappa^2}{\mu^2}~\!,$$
and the linear equation \eqref{psiconstraint} can be cast into the form
\begin{equation}
\psi_{mn} = \dfrac{1}{\mu^2-\kappa^2} \big( \mu^2 \psi - \fie_{mn} \chi\big)~\!.\label{psimn-bis}
\end{equation}
However, as was shown in the same Lemma, $\psi_{mn}$ also satisfies equation \eqref{psimn}, which we shall rewrite as
\begin{equation*}
\psi_{mn} = \dfrac{\psi_m + \psi_n}{2 U}~\!,
\end{equation*}
in terms of the (new) variable 
\begin{equation}
U := \dfrac{\tau \tau_{mn}}{\tau_m\tau_n}~\!.
\end{equation}
This of course gives a simple expression for the sum of $\psi_m$ and $\psi_n$:
$$\psi_m + \psi_n = \dfrac{2U}{\mu^2-\kappa^2} \big( \mu^2 \psi - \fie_{mn} \chi\big)~\!.$$
It then suffices to define an auxiliary function $\psi'$ 
\begin{equation}
\psi' := \dfrac{\mu^2-\kappa^2}{2\mu} (\psi_m - \psi_n)~\!,\label{psipdef}
\end{equation}
to obtain the linear equations
\begin{gather*}
\psi_m = \dfrac{1}{\mu^2-\kappa^2} \big( \mu^2 U \psi + \mu \psi' - U \fie_{mn} \chi\big)~\!,\\
\psi_n = \dfrac{1}{\mu^2-\kappa^2} \big( \mu^2 U \psi - \mu \psi' - U \fie_{mn} \chi\big)~\!,
\end{gather*}
linking all three functions $\psi, \psi'$ and $\chi$. Repeated use of equations \eqref{psimn-bis} and \eqref{chidefeq} in both the $m$ and $n$ directions then yields the following linear system for the function $\Psi:= ~\!^t\ \!\!( \psi~\!~\psi'~\!~\chi )$:
\begin{gather}
\Psi_m = \dfrac{1}{\mu^2-\kappa^2}~\begin{pmatrix}
\mu^2 U & \mu & -U \fie_{mn}\\[-3mm]\\
\mu~\! A & U_m \big(\mu^2-\frac{1}{\mu} \fie_{mmn}\phi^*\big) & \frac{1}{\mu} B\\[-3mm]\\
\mu~\! U \phi^* & \phi^* & (\mu^2-\kappa^2) - \frac{1}{\mu}U \fie_{mn} \phi^*
\end{pmatrix} \cdot \Psi~\!,\label{dYO-LP1}\\
\Psi_n = \dfrac{1}{\mu^2-\kappa^2}~\begin{pmatrix}
\mu^2 U & -\mu & -U \fie_{mn}\\[-3mm]\\
\mu~\! C & U_n \big(\mu^2+ \frac{1}{\mu} \fie_{mnn}\phi^*\big) & \frac{1}{\mu} D\\[-3mm]\\
-\mu U \phi^* & \phi^* & (\mu^2-\kappa^2) + \frac{1}{\mu} U \fie_{mn} \phi^*
\end{pmatrix} \cdot \Psi~\!,\label{dYO-LP2}
\end{gather}
with
\begin{gather*}
A = - (\mu^2-\kappa^2) + U U_m \big(\mu^2 - \frac{1}{\mu} \fie_{mmn}\phi^*\big)\\
B = (\mu^2-\kappa^2) \big(\fie_{mn} - U_m \fie_{mmn}\big) - U U_m \big(\mu^2 - \frac{1}{\mu} \fie_{mmn}\phi^*\big) \fie_{mn} \\
C = (\mu^2-\kappa^2) - U U_n \big(\mu^2 + \frac{1}{\mu} \fie_{mnn}\phi^*\big)\\
D = (\mu^2-\kappa^2)  \big(-\fie_{mn} + U_n \fie_{mnn}\big) + U U_n \big(\mu^2+ \frac{1}{\mu} \fie_{mnn}\phi^*\big) \fie_{mn}~\!.
\end{gather*}
This system is compatible, independently of the value of $\kappa$, provided that the functions $U, \phi^*$ and $\fie$ satisfy the discrete Yajima-Oikawa system \eqref{dYO}. Hence, (\ref{dYO-LP1},\ref{dYO-LP2}) constitutes a Lax pair with spectral parameter $\kappa$ for the discrete Yajima-Oikawa system.
\end{proof}
The following Lax pair for the continuous Yajima-Oikawa system \eqref{YO} can be obtained at the continuum limit of the discrete Lax pair (\ref{dYO-LP1},\ref{dYO-LP2}), with the ansatz $U = 1 - u/\mu^2$ for the function $U$:
\begin{gather*}
\Psi_{x_1} = \begin{pmatrix} 0 & 1 & 0 \\
\kappa^2-2 u & 0 & - \fie\\
\phi^* & 0 & 0 \end{pmatrix}\cdot\Psi~\!,\\[-3mm]\\
\Psi_{x_2} = \begin{pmatrix} \kappa^2 & 0 & -\fie \\
-\fie\phi^* & \kappa^2 & - \fie_{x_1}\\
\phi_{x_1}^* & \phi^* & 0 \end{pmatrix}\cdot\Psi~\!.
\end{gather*}
Note that the continuum limit of $\psi'$ as given by \eqref{psipdef} is just $\psi_{x_1}$, which is nothing but the first relation in the Lax pair. Since the (continuous function) $\chi$ can be regarded as the squared eigenfunction potential $\Omega(\psi, \phi^*)$ defined for $\psi$ and $\phi^*$, this equation can be interpreted as $\psi_{x_1x_1} = (\kappa^2-2 u)\psi - \fie~\! \Omega(\psi,\phi^*)$ which gives a specific realization of the pseudo-differential constraint $\big( \partial_{x_1}^2 + 2 u + \fie \partial^{-1} \phi^* \big) \psi = \kappa^2 \psi$ discussed in \cite{Cheng}.

\section{Hierarchies and extensions}\label{extensions}
In this section we shall discuss two possible extensions of our approach to discrete constraints. A first extension concerns discretisations for entire hierarchies of integrable systems. Using the example of the discrete NLS equation we shall argue that we have, in fact, already obtained a discrete form for the entire (continuous) NLS hierarchy. A second extension is that to discrete constraints involving shift operators in more than 3 dimensions. It will be explained that such constraints can indeed yield discretisations of higher order constrained KP hierarchies such as that associated to the Melnikov system \cite{Melnikov}, but that problems arise when one tries to construct Lax pairs for single equations contained in the corresponding hierarchies. 

\smallskip
As was explained in section \ref{mainresults}, a first, simple continuum limit (e.g. $|\nu|\rightarrow\infty$) of the discrete NLS equation \eqref{dNLS}, yields the semi-discrete system \eqref{sdNLS}
\begin{equation}
\left\{\begin{array}{l}
\mu(\phi_m^* - \phi^*) = \phi_{x_1}^* - \frac{1}{\mu} (\phi^*)^2 \fie_m\\[-2mm]\\
\mu(\fie_m - \fie) = (\fie_m)_{x_1} + \frac{1}{\mu} \phi^* (\fie_m)^2~\!.
\end{array}\right.\label{sdNLS-bis}
\end{equation}
As explained before, a subsequent limit, $|\mu|\rightarrow\infty$, yields the NLS equation, at lowest order.

It is however interesting to take a closer look at the general expansion of the equations in this system, rather than just considering the lowest order part. In general, for $|\mu|\approx \infty$, one obtains
\begin{equation}
\left\{\begin{array}{l}
\!\!\!\!\sum_{j=2}^{+\infty} ~\!\mu^{1-j} p_j(\tilde\partial) \phi^* = - (\phi^*)^2 \sum_{j=0}^{+\infty}~\! \mu^{-1-j} p_j(\tilde\partial) \fie\\[-2mm]\\
\!\!\!\!\sum_{j=2}^{+\infty} ~\!\mu^{1-j} p_j(\tilde\partial) \fie = \left( \sum_{j=1}^{+\infty}~\! \mu^{-j} p_j(\tilde\partial) \fie\right)_{\!x_1}\!\! + \mu^{-1}~\! \phi^* \!\left( \sum_{j=0}^{+\infty}~\! \mu^{-j} p_j(\tilde\partial) \fie\right)^2
\end{array}\right.\label{NLS-expansion}
\end{equation}
where, as in the case of \eqref{KP-LP} and \eqref{KP-aLP}, $p_j(\tilde\partial)$ denotes the weight $j$ Schur polynomial in the (weighted) partial differential operators $\tilde\partial=( \partial_{x_1},\partial_{x_2}/2, \partial_{x_3}/3, \cdots)$. The NLS equation \eqref{NLS} is obtained from these expansions at order $\mu^{-1}$. Supposing that these expansions are satisfied identically, i.e. that they are satisfied at {\em all} orders in $\mu$, one obtains the recursive system:
\begin{equation}
\left\{\begin{array}{l}
p_j(\tilde\partial) \phi^* = -(\phi^*)^2~\! p_{j-2} (\tilde\partial) \fie\\[-2mm]\\
p_j(\tilde\partial) \fie = p_{j-1}(\tilde\partial) \fie_{x_1} + \phi^* \sum_{i=0}^{j-2} \big(p_i(\tilde\partial) \fie\big) \big(p_{j-2-i}(\tilde\partial) \fie\big)
\end{array}\right.\quad(~\!^\forall j \geq 2)~\!.\label{NLS-recur1}
\end{equation}
For example, at $j=3$ one finds the system
\begin{equation*}
\left\{\begin{array}{l}
\phi_{x_3}^* = \phi_{3 x_1}^* + 6  \phi^* \phi_{x_1}^* \fie\\[-2mm]\\
\fie_{x_3} = \fie_{3 x_1} + 6 \phi^* \fie \fie_{x_1}
\end{array}\right.~\!,
\end{equation*}
if one uses the NLS equation \eqref{NLS} to eliminate all $x_2$ derivatives. This system is readily identified as the first higher order flow in the NLS hierarchy, generated by the recursion relation \cite{Cheng}:
$$\begin{pmatrix} \fie \\ \phi^*\end{pmatrix}_{x_r} = \begin{pmatrix} \partial_{x_1} + 2 \fie \partial^{-1} \phi^* & 2 \fie\partial^{-1} \fie \\
-2 \phi^* \partial^{-1} \phi^* & -\partial_{x_1}-2 \phi^* \partial^{-1} \fie
\end{pmatrix}^r\cdot \begin{pmatrix}\fie \\ -\phi^*\end{pmatrix}\quad (~\!^\forall r\geq 2)~\!,$$
where $\partial^{-1}$ denotes the formal inverse of $\partial_{x_1}$ : $\partial_{x_1}\partial^{-1} = 1 =  \partial^{-1}\partial_{x_1}$.

In fact, the nonlinear recursion \eqref{NLS-recur1} does yield the entire hierarchy of commuting flows associated to the NLS equation (apart from the, trivial, $x_1$ flow). Or, in other words, the expansion \eqref{NLS-expansion} is indeed valid at all orders. This can be understood quite easily if one remembers how the discrete NLS equation \eqref{dNLS} was obtained. As was explained in the proof of Proposition \ref{dNLSprop} in section \ref{proofs}, the $m$ direction that appears in de discrete NLS equation is just one choice, among infinitely many possible directions on the discrete KP lattice. One could have performed the same calculations as those for the $m$ direction, for each of these other directions, each time obtaining exactly the same equation \eqref{dNLS}, but for the shifts in the $m$ direction that need to be replaced by shifts in the newly chosen direction. Since the constraint was shown to be consistent with the entire discrete KP hierarchy, it turns out that there is a single discrete equation  which is valid in infinitely many directions (e.g. $\ell_i$, with lattice parameter $a_i$):
\begin{equation*}
\left\{\begin{matrix}
a_i \phi_{\ell_i}^* - \nu \phi_n^* = \dfrac{(a_i-\nu) \phi^*}{1-\frac{1}{\nu a_i} \phi^*\fie_{n\ell_i}}~\!\\[-1.5mm]\\
a_i \fie_n - \nu \fie_{\ell_i} = \dfrac{(a_i-\nu) \fie_{n\ell_i}}{1-\frac{1}{\nu a_i} \phi^*\fie_{n \ell_i}}~\!.
\end{matrix}\right.
\end{equation*}
Infinitely many copies of this equation, when considered simultaneously (so as to make the expansion \eqref{NLS-expansion} valid at all orders) then offer a discretisation of the entire NLS hierarchy. Needless to say, this situation faithfully mimics that of the discrete KP hierarchy. This then leads to the proposition:
\begin{prop}
The equations \eqref{dNLS}, \eqref{dBK} and \eqref{dYO} can be regarded as  discretisations of the hierarchies of commuting flows associated to, respectively,  the NLS, Broer-Kaup and Yajima-Oikawa equations.
\end{prop}
A similar property has been reported in \cite{Nijhoff-LGD} for the lattice KdV and lattice Boussinesq equations (cf. also \cite{TongasN} for some intriguing connections to other types of integrable equations).

We shall not present the precise forms of the recursion formulae for the Broer-Kaup or Yajima-Oikawa hierarchies, as the relations become quite cumbersome, but it is interesting to note that, instead of using the equation for $\fie$ in \eqref{sdNLS-bis} to calculate the continuum limit, one could just as well have started from its down-shifted version
$$\mu(\fie - \fie_{m'}) = \fie_{x_1} + \frac{1}{\mu} \phi_{m'}^* (\fie)^2~\!,$$
which then results in a different and much simpler recursive formula for the $\fie$-flows:
\begin{equation}
p_j({-\tilde\partial})\fie = - \fie^2~\! p_{j-2}(-\tilde\partial) \phi^*\quad (~\!^\forall j\geq 2)~\!.\label{NLS-recur2}
\end{equation}

Furthermore, as explained in the proof of proposition \ref{dNLSprop}, the continuum limit of the discrete constraint can be written as 
$$\big(\log\tau\big)_{x_1} =  \Omega(\fie,\phi^*)~\!,$$
which, combined with the fact that (the continuum limits) of $\fie$ and $\phi^*$ satisfy the KP linear system \eqref{KP-LP} and adjoint linear problem \eqref{KP-aLP}, then leads to yet another recursive formulation of the NLS hierarchy, in terms of the squared eigenfunction potential $\Omega(\fie,\phi^*)$:
\begin{equation*}
\left\{\begin{array}{l}
p_j(\tilde\partial) \phi^* = -\phi^*~\! p_{j-1}(\tilde\partial) \Omega(\fie,\phi^*)\\[-2mm]\\
p_j(-\tilde\partial) \fie= \fie~\! p_{j-1}(-\tilde\partial) \Omega(\fie,\phi^*)
\end{array}\right.\quad(~\!^\forall j \geq 2)~\!.\label{NLS-recur3}
\end{equation*}
Combining these expressions with those in \eqref{NLS-recur1} and \eqref{NLS-recur2}, one obtains the fundamental identities
$$^\forall j\geq 2:\quad p_{j-1}(\tilde\partial) \Omega(\fie,\phi^*) = \phi^*~\!p_{j-2} (\tilde\partial) \fie~\!,\quad p_{j-1}(-\tilde\partial) \Omega(\fie,\phi^*)= -\fie~\!p_{j-2}(-\tilde\partial) \phi^*~\!,$$
for the squared eigenfunction potential that were shown to hold for the case of the (unconstrained) KP hierarchy in \cite{RW-IP}.

\smallskip
Another interesting problem is that of constructing higher order constraints which would offer integrable discretisations of general $k$-constrained hierarchies. Let us consider the case of the 3-constrained KP hierarchy, which yields the famous Melnikov system \cite{Melnikov}, as an example.

Let us impose the $3$-constraint
 \begin{equation}
 \tau_{x_3} = \tau~\!\Omega(\phi,\phi^*)\label{3constraint}
 \end{equation}
on a KP tau function -- i.e. $\tau$ satisfies \eqref{KP-BF} -- and a KP eigenfunction $\phi$ and adjoint eigenfunction $\phi^*$. In other words, $\phi$ and $\phi^*$ satisfy the KP linear system (and its adjoint), the first few equations of which are:
\begin{gather}
\left\{\begin{array}{l}
\phi_{x_2} = \phi_{2 x_1} + 2 u \phi\\[-3mm]\\
\phi_{x_2}^* = - \big(\phi_{2 x_1}^* + 2 u \phi^*\big)
\end{array}\right.\label{KP-ZS2}
\end{gather} 
and
\begin{gather*}
\left\{\begin{array}{l}
\phi_{x_3} = \phi_{3 x_1} + 3 u \phi_{x_1} + \frac{3}{2} (u_{x_1}+v) \phi\\[-3mm]\\
\phi_{x_3}^* = \phi_{3 x_1}^* + 3 u \phi_{x_1}^* + \frac{3}{2} (u_{x_1}-v) \phi^*~\!.
\end{array}\right.\label{KP-ZS3}
\end{gather*} 
These two sets of equations are compatible if and only if $u$ and $v$ satisfy the KP equation in the form,
\begin{equation}
u_{x_2} = v_{x_1}~\!,\qquad u_{x_3} = \frac{1}{4} (u_{3x_1} + 12 u u_{x_1}) + \frac{3}{4} v_{x_2}~\!,\label{KP}
\end{equation}
the first relation of which is trivial if one parametrizes $u$ and $v$ in terms of the tau functions as $u = \big(\log\tau\big)_{2 x_1}$ and $v = \big(\log\tau\big)_{x_1 x_2}~\!.$
The 3-constraint then implies
$$u_{x_3} = \big(\Omega(\phi,\phi^*)\big)_{2x_1} = \big(\phi\phi^*\big)_{x_1}\quad\text{and}\quad v_{x_3} = \big(\Omega(\phi,\phi^*)\big)_{x_1 x_2} = \big(\phi_{x_1}\phi^*- \phi\phi_{x_1}^*\big)_{x_1}~\!,$$ 
and therefore immediately yields the third order coupled system
\begin{equation}
\left\{\begin{array}{l}
u_{x_3} = \big(\phi\phi^*\big)_{\!x_1}\\
v_{x_3} = \big(\phi_{x_1}\phi^*- \phi\phi_{x_1}^*\big)_{\!x_1}\\[-3mm]\\
\phi_{x_3} = \phi_{3 x_1} + 3 u \phi_{x_1} + \frac{3}{2} (u_{x_1}+v) \phi\\[-3mm]\\
\phi_{x_3}^* = \phi_{3 x_1}^* + 3 u \phi_{x_1}^* + \frac{3}{2} (u_{x_1}-v) \phi^*~\!,
\end{array}\right.\label{Melnikov3}
\end{equation}
which is however {\em not} the lowest member in its associated hierarchy of commuting flows. The lowest member is actually obtained from \eqref{KP-ZS2} by imposing the KP equation \eqref{KP} under the 3-constraint:
\begin{gather}
\left\{\begin{array}{l}
u_{x_2} = v_{x_1}\\[-4mm]\\
3 v_{x_2} + u_{3x_1} + 12 u u_{x_1}  = 4 \big(\phi \phi^*\big)_{x_1}\\[-4mm]\\
\phi_{x_2} = \phi_{2 x_1} + 2 u \phi\\[-3mm]\\
\phi_{x_2}^* = - \big(\phi_{2 x_1}^* + 2 u \phi^*\big)~\!,
\end{array}\right.\label{Melnikov2}
\end{gather} 
which is the Melnikov system \cite{Melnikov}. Because of the form of the nonlinear equations for $u$ and $v$, this system is often referred to as the ``Boussinesq equation with sources''. Although this is an integrable system in its own right, with a Lax pair, bi-Hamiltonian structure \cite{Oevel-etal} etc., from the above construction it seems quite natural for it to appear alongside the next member in the hierarchy, system \eqref{Melnikov3}. This situation is typical for all higher order constrained systems, i.e. beyond the Yajima-Oikawa system. As we will see next, this also seriously complicates the construction of any isolated discrete counterparts.

It is quite easy to convince oneself that a discrete constraint \eqref{constraint} that only involves two (or less) lattice directions, can never yield the 3-constraint \eqref{3constraint} at the continuum limit. The first possibility therefore arises for the choice $S= T_l T_m T_n$, involving all three lattice directions in a single HM equation \eqref{HMs}. However, as in the case of the Yajima-Oikawa reduction,  the lattice parameters in those three directions ($\lambda, \mu$ and $\nu$) have to be restricted as well, in order to find the correct limit. In short, restricting the lattice parameters by
$$\lambda = \omega^2 \mu~\!, \quad \nu = \omega \mu~\!,\qquad \text{for~} \omega\in\mathbb{C}:\ \  \omega^2+\omega + 1=0~\!,$$
and choosing $\gamma= \lambda\mu\nu=\mu^3$, the constraint \eqref{constraint} becomes
\begin{equation}
\mu^3 (\tau_{lmn}-\tau) = \tau~\! \Omega(\phi,\phi^*)~\!,\label{Melconstr}
\end{equation}
which indeed yields \eqref{3constraint} at the limit $|\mu|\rightarrow\infty$. As before, the functions $\phi$ and $\phi^*$ are required to satisfy the equations \eqref{phidefeq} and \eqref{phistdefeq} in all possible lattice directions.

It is now straightforward to construct a discrete system that, at its continuum limit, will generate the Melnikov system \eqref{Melnikov2}. Eliminating all $l$ shifts in \eqref{HMs} by means of the constraint \eqref{Melconstr}, one obtains the system
\begin{equation}
\left\{\begin{array}{l}
\phi_{mn} = \dfrac{1}{1-\omega}~\!\dfrac{\tau_m\tau_n}{\tau\tau_{mn}} \big(\phi_n - \omega\phi_m\big)\\[-2mm]\\
\phi^* = \dfrac{1}{1-\omega}~\!\dfrac{\tau_m\tau_n}{\tau\tau_{mn}} \big(\phi_m^* - \omega\phi_n^*\big)\\[-2mm]\\
\tau_{mn}\tau_{m'n'}~\! \Omega_{m'n'}^+~\! + ~\! \omega~\! \tau_m\tau_{m'}~\! \Omega_{m'}^+  ~\!  =~\!  (1+\omega)~\! \tau_n\tau_{n'}~\! \Omega_{n'}^+ 
\end{array}\right.\label{dMel2}
\end{equation}
where $\Omega^+ := \Omega(\phi,\phi^*) + \mu^3$. Note that the only reason for introducing this explicit shift $\mu^3$ in the squared eigenfunctions, is that it is the correct ansatz for the continuum limit. Taking $|\mu|\rightarrow\infty$, we obtain the Melnikov system \eqref{Melnikov2} but with the equations for $u$ and $v$ replaced by the single equation
$$\big(3 D_{x_2}^2 + D_{x_1}^4\big) \tau\cdot\tau = 8 \bar\tau \tau^*~\!,$$
when $\phi$ and $\phi^*$ are parametrized as in \eqref{phietcdef}. This equation is in fact the potential form of the equations for $u$ and $v$ in \eqref{Melnikov2}. Hence, although it is very closely related to the Melnikov system and although it can be argued that it is an integrable system as well, it is clear that the system \eqref{dMel2} cannot be obtained as the compatibility condition of a Lax pair. To obtain a genuine discretisation of \eqref{Melnikov2} one would need to take exactly the correct discrete `derivative' of the equation for $\tau$ and $\Omega^+$, possibly after implementing a clever dependent variable transformation that would eliminate the latter variable from the equation. However, even if one succeeds in doing so, this will not eliminate all problems. The biggest hurdle to obtaining an integrable discretisation of the Melnikov system is the construction of its Lax pair. Note that the system \eqref{dMel2} was expressed in $\phi$ and not in the function $\fie$ of Lemma \ref{crucial lemma}. This is to eliminate all occurrences of the $l$ direction, which would have resurfaced through the shift operator $S$ had we used $\fie$ in the equations. Similarly, from the construction of the Lax pair, and especially because of \eqref{psiconstraint} in Lemma \ref{LPfromHM}, it is clear that it is extremely difficult to eliminate the $l$ direction from the Lax pair and obtain equations defined on a 2 dimensional lattice.
One way to circumvent this problem would be to give up the ambition of constructing a discretisation of just \eqref{Melnikov2}, and to aim for a discrete system that contains both \eqref{Melnikov2} and \eqref{Melnikov3} at its continuum limit. The following, rather simple, system is in fact just that, 
\begin{equation}
\left\{\begin{array}{l}
\phi_{mn} = \dfrac{1}{1-\omega}~\!\dfrac{\tau_m\tau_n}{\tau\tau_{mn}} \big(\phi_n - \omega\phi_m\big)\\[-2mm]\\
\phi_{lm} = \dfrac{1}{1-\omega^2}~\!\dfrac{\tau_l\tau_m}{\tau\tau_{lm}} \big(\phi_l-\omega^2\phi_m\big)\\[-2mm]\\
\phi^* = \dfrac{1}{1-\omega}~\!\dfrac{\tau_m\tau_n}{\tau\tau_{mn}} \big(\phi_m^* - \omega\phi_n^*\big)~=~\! \dfrac{1}{1-\omega^2}~\!\dfrac{\tau_l\tau_m}{\tau\tau_{lm}} \big(\phi_m^*-\omega^2\phi_l^* \big)\\[-2mm]\\
\mu^3 \Delta_m \dfrac{\tau_{lmn}}{\tau} = \phi^*\phi_m
\end{array}\right.\label{dMel23}
\end{equation}
albeit for the potential forms of the equations for $u$ and $v$ that appear in these systems.

Note that also in the discrete case, the Melnikov system \eqref{dMel2} is related to the Boussinesq system. If we impose $\phi=\phi^*=0, ~\!\Omega^+=\mu^3$ on this system, we obtain
\begin{equation}
\tau_{mn}\tau_{m'n'}~\! + ~\! \omega~\! \tau_m\tau_{m'}~\! =~\!  (1+\omega)~\! \tau_n\tau_{n'}~\!,\label{dBsq-BF}
\end{equation}
which is the Hirota bilinear form of the discrete Boussinesq equation \cite{DateIII}. Introducing the variable $W := \tau/\tau_{n'}$, this bilinear form can be shown to be related to the equation
\begin{equation}
\dfrac{W_{mn}}{W_n} - \dfrac{W_{n'}}{W_{m'n'}} = \omega \Big(\dfrac{W_{n'}}{W_m} - \dfrac{W_{m'}}{W_{n}}\Big)~\!,\label{dBsq}
\end{equation}
which, through the ansatz $W = 1 + (\omega^2/\mu)~\! w + {\mathcal O}(1/\mu^2)$, contains the Boussinesq equation
$$w_{2x_2} + w_{4x_1} + 12 w_{x_1} w_{2x_1} =0$$
at its continuum limit. (Compare this equation to the $u,v$ part of \eqref{Melnikov2} with $u=w_{x_1}$ and $v=w_{x_2}$.) 
The system \eqref{dBsq} arises as the compatibility condition of the Lax pair
\begin{equation*}
\left\{\begin{array}{l}
\psi_{m'} = \omega \big( \psi_{n'} + (1-\omega)  \frac{W}{W_{m'}} \psi\big)\\[-2mm]\\
\psi_{n'n'n'} + (\omega-1) \big( \frac{W_{m'n'}}{W} +(1+\omega) \frac{W_{n'n'}}{W_{m'n'n'}}\big) \psi_{n'n'} + 3 \frac{W_{n'}}{W} \psi_{n'} = \kappa \psi~\!,
\end{array}\right.
\end{equation*}
which has the usual third order Lax pair for the Boussinesq equation
\begin{equation*}
\left\{\begin{array}{l}
\psi_{x_2} = \psi_{2x_1} + 2 w_{x_1} \psi\\[-2mm]\\
\psi_{3x_1} + 3 w_{x_1} \psi_{x_1} + \dfrac{3}{2} (w_{2x_1} + w_{x_2}) \psi = K \psi~\!,
\end{array}\right.
\end{equation*}
 as its continuum limit.

Under the above condition, $\phi=\phi^*=0, ~\!\Omega^+=\mu^3$, the constraint reduces to $\tau_{lmn} = \tau$, which is still expressible on a single HM equation, hence the bilinear equation \eqref{dBsq-BF}. Moreover, it is exactly the type of quasi-periodic reduction of the discrete KP hierarchy that is known to yield discretisations of continuous systems that are related to principal realizations of $A_n^{(1)}$-type, when accompanied by the appropriate restrictions on the lattice parameters. Reductions of the HM equation that do not respect these conditions on the parameters are also known \cite{MarunoK}, but it is not clear how the Miwa transformation relates these systems to their continuous counterparts and, as a result, their associated symmetry algebras are not understood.  In this case however, since the continuum limit goes through without any problems, one can conclude that the system is associated to the $A_2^{(1)}$ algebra. The original Melnikov system, however, is known to be related to a realization of the $A_3^{(1)}$ algebra \cite{RW-JMP2}. Hence, although the equations themselves appear to be related, it is hard to imagine that their solution spaces will have much in common. In fact, for the continuous Boussinesq equation, the only known class of solutions that can be generated using bi-directional Wronskians such as those that arise for the Melnikov system, consists essentially of polynomials. The soliton solutions in particular have a different structure and the exact nature of the interrelation of the Melnikov and Boussinesq equations, in both the discrete and continuous realms, remains an interesting open problem.

\section{Prospects}\label{discussion}
In the preceding sections we explained how to construct discrete integrable versions of the systems obtained through so-called symmetry reductions of the KP hierarchy. We succeeded in obtaining not only the integrable discretisations of the NLS, Broer-Kaup and Yajima-Oikawa equations, but in fact also of their entire associated hierarchies. It is an interesting open problem to try to understand the recursive formulations of the continuous hierarchies that are obtained from these discretisations, in relation to the recursion operators or bi-Hamiltonian structures that generate the continuous hierarchies. It seems highly unlikely that no explicit link between these two descriptions exists.

The discretisations of the NLS and Yajima-Oikawa equations are particularly interesting because of the known links of their continuous counterparts to Painlev\'e equations. In fact, it might be interesting to try to find non-autonomous extensions of these discrete systems, and to study their reductions to non-autonomous mappings. It is worth pointing out however that the constraint-based construction that led to these discrete systems, does not seem to be applicable to the case of a non-autonomous lattice. The most immediate (and most important problem) being that there does not seem to exist a non-autonomous version of Lemma \ref{crucial lemma}.

All the discrete systems constructed here correspond to  $A_n^{(1)}$-type symmetry algebras, albeit in realizations that (apart from the NLS-case which was obtained in \cite{DateIV}) had not been studied in a discrete context,  until now. However, there exists an intriguing escape route that might make it possible to construct at least one system with a completely different underlying symmetry algebra. Just as  with \eqref{Melnikov2} and the discrete Boussinesq equation, there exists a (further) reduction of \eqref{Melnikov3} to a system associated to a completely different symmetry algebra. The `self-adjoint' reduction $\phi=\phi^*, v=0$ of the system \eqref{Melnikov3}
\begin{equation*}
\left\{\begin{array}{l}
u_{x_3} = \big(\phi^2\big)_{\!x_1}\\[-3mm]\\
\phi_{x_3} = \phi_{3 x_1} + 3 u \phi_{x_1} + \frac{3}{2} u_{x_1} \phi
\end{array}\right.
\end{equation*}
 is in fact known \cite{SidorenkoS} to be related to the $D_3^{(2)}$ algebra \cite{Wilson} . If it would be possible to implement this reduction on the discrete system \eqref{dMel23} (or on a more sophisticated version of it), one would obtain what might be the first example of an integrable additive (2 dimensional) discretisation of a system associated with an affine algebra that is not of $A_n^{(1)}$ or $A_{2n}^{(2)}$ type. Such a reduction on a discrete system must however be quite subtle, as requiring self-adjointness most often entails the identification of several directions on the lattice, which is to be avoided at all cost.

\appendix
\section{\bf Some facts concerning the discrete KP hierarchy}\label{app}
In this appendix it will be shown that the discrete KP hierarchy of Y. Ohta et al., defined in \cite{Ohta-dKP}, is generated by a rather small set of coupled Hirota-Miwa equations. The notation used in the proof of this statement differs slightly from that in the main body of the paper since we shall need an arbitrary number of lattice directions. Let us define tau functions $\tau : ~~ \mathbb{Z}^n~ \longrightarrow~ \mathbb{C}$ on an $n$ dimensional lattice, the directions on which will be denoted as $\ell_i~(i=1, 2, \hdots, n)$. We shall write $\tau(\ell_1, \hdots, \ell_n)$ for the value the tau function takes at the lattice site $(\ell_1, \hdots, \ell_n)$. A shift operation on this $n$ dimensional lattice, say in the direction $\ell_j$, will be denoted as $\s_j[\cdot]$ :
$\s_j[\tau] := \tau(\ell_1,\hdots, \ell_{j-1}, \ell_j+1, \ell_{j+1}, \hdots, \ell_n)$ or, alternatively, by a subscript to the symbol $\tau$ : $ \s_j[\tau] \equiv \tau_j$ (and hence, $\tau_{ij}$ will denote $\big(\s_i\s_j\big)[\tau], etc.$).
Successive shifts will be denoted as
$$ \ti{j}{j+\kappa} :=~\Big(\prod_{k=0}^\kappa \s_{j+k}\Big)[\tau]~\!,$$
and successive shifts with a shift in a single direction lacking as:
$$\tih{j\cdd}{i}{\cdd j+\kappa}:=~\Big(\prod_{\stackrel{k=0}{_{j+k\neq i}}}^\kappa \s_{j+k}\Big)[\tau]\qquad{ (j\leq i\leq j+\kappa)}~\!.$$

Similarly, $\Vi{j}{j+\kappa}$ denotes the Vandermonde determinant 
$$\Vi{j}{j+\kappa} := \prod_{\stackrel{k,m=j}{_{k>m}}}^{j+\kappa} (a_k-a_m)~,$$
and $\Vih{j\cdd}{i}{\cdd j+\kappa}$ denotes a Vandermonde determinant with one column (and its last row) missing:
$$\Vih{j\cdd}{i}{\cdd j+\kappa} := \prod_{\stackrel{k,m=j}{_{k,m\neq i, k>m}}}^{j+\kappa} (a_k-a_m)\qquad { (j\leq i\leq j+\kappa)}~.$$

Finally, by $E_{i j k}$ ($i<j<k$) we denote the Hirota-Miwa equation:
$$(a_k-a_j) \tau_i \tau_{jk} + (a_i-a_k) \tau_j \tau_{ik} + (a_j-a_i) \tau_k \tau_{ij} = 0~\!,$$
where $a_i, a_j, a_k$ are the (reciprocals of) the lattice parameters in the directions $\ell_i, \ell_j$ and $\ell_k$, respectively.

In \cite{Ohta-dKP} it is shown that for increasing $n\geq3$, the $n$ term bilinear relation 
\begin{gather}
\begin{vmatrix}
\tau_1~\! \ti{2}{n}& \tau_2~\! \tih{1}{2}{\cdd n} &\cdots & \tau_{n-1}~\!\tih{1\cdd}{n-1}{n} & \tau_n~\!\ti{1}{n-1}\\
1 & 1 & \cdots & 1 & 1\\
a_1 & a_2 &\cdots & a_{n-1} & a_2\\
\vdots & \vdots& \ddots & \vdots & \vdots\\
a_1^{n-2} & a_2^{n-2} & \cdots & a_{n-1}^{n-2} & a_n^{n-2}
\end{vmatrix} = 0~\!,\label{dbKP}
\end{gather}
generates, at the continuum limit, the entire bilinear KP hierarchy and can therefore be regarded as its discrete version. Relation \eqref{dbKP}, for arbitrary $n$, is therefore commonly referred to as the {\em discrete KP hierarchy}. Note that for $n=3$, the determinant in \eqref{dbKP} yields  nothing but the Hirota-Miwa equation $E_{1 2 3}$. However, the above definition of the discrete KP hierarchy is by no means elementary and in fact, the following proposition can be proven:

\begin{prop} The set of Hirota-Miwa equations
$$\mathbb{E} = \Big\{~\! E_{i j k}~\!\Big|~\! i,j,k,\in\{1,2, \hdots, n\},~ i<j<k~\!\Big\}$$
generates the $n$ term relation \eqref{dbKP}.
\end{prop}

\begin{proof} First of all, note that the coefficient of the `$j$th' term in \eqref{dbKP}, i.e. the term $\tau_j~\!\tih{1\cdd}{j}{\cdd n}~\!$, is equal to $(-1)^{j+1} \Vih{1\cdd}{j}{\cdd n}$. Next, for $n$ sufficiently large, we consider two different $n-1$ term versions of the l.h.s. of \eqref{dbKP}, defined on the $n-1$ dimensional sub-lattices generated by the $\ell_1, \ell_2, \hdots, \ell_{n-1}$ and the $\ell_2, \ell_3, \hdots, \ell_{n}$ directions respectively. We shall denote these expressions as $E_{1\cdd n-1}$ and $E_{2\cdd n}$ :
\begin{gather*}
\begin{split}
E_{1\cdd n-1}~:\qquad \Vi{2}{n-1}~\! \tau_1~\! \ti{2}{n-1} - \Vih{1}{2}{\cdd n-1}~\! \tau_2~\! \tih{1 }{2}{\cdd n-1} \qquad\qquad\qquad\qquad\qquad \\ +\ \cdots\ + (-1)^n \Vi{1}{n-2}~\! \tau_{n-1} ~\!\ti{1}{n-2}\end{split}\\[-3mm]\\
E_{2\cdd n}~:\qquad \Vi{3}{n}~\! \tau_2~\! \ti{3}{n} - \Vih{2}{3}{\cdd n}~\! \tau_3~\! \tih{2 }{3}{\cdd n} + \cdots + (-1)^n \Vi{2}{n-1}~\! \tau_{n} ~\!\ti{2}{n-1}
\end{gather*}
For these expressions, we calculate the linear combination
\begin{gather}
\left(\prod_{k=2}^{n-1} (a_n-a_k)\right)~\!\tau_1~ \s_n[E_{1\cdd n-1}]  ~~\! - ~ \left(\prod_{k=2}^{n-1} (a_k-a_1)\right)~\!\tau_n~ \s_1[E_{2\cdd n}]~\!.\label{lc}
\end{gather}

The `$1$st' term, i.e. the term in $\tau_1 \tau_{1 n} \ti{2}{n}$, in the resulting expression for \eqref{lc} has 
\begin{equation}
\Vi{2}{n-1} \left(\prod_{k=2}^{n-1} (a_n-a_k)\right) = \Vi{2}{n}\label{cof1}
\end{equation}
as a coefficient, whereas the last term, i.e. that in $\tau_n~\! \tau_{1 n}~\! \ti{1}{n-1}$, has coefficient
\begin{equation}
-  \left(\prod_{k=2}^{n-1} (a_k-a_1)\right) (-1)^n \Vi{2}{n-1} = (-1)^{n+1} \Vi{1}{n-1}~.\label{cofn}
\end{equation}
Moreover, for $j :~ 2\leq j\leq n-1$, the term in $\tau_j~\! \tih{1\cdd}{j}{\cdd n-1}$ in $E_{1\cdd n-1}$, combined with that in $\tau_j~\! \tih{2\cdd}{j}{\cdd n}$ in $E_{2\cdd n}$, gives rise to a contribution
\begin{multline}
(-1)^{j+1} \Vih{1\cdd}{j}{\cdd n-1} \left(\prod_{k=2}^{n-1} (a_n-a_k)\right) \tau_1~\! \tau_{j n}~\! \tih{1\cdd}{j}{\cdd n} \\
- (-1)^j \Vih{2\cdd}{j}{\cdd n}\left(\prod_{k=2}^{n-1} (a_k-a_1)\right) \tau_n~\!\tau_{1 j}~\!\tih{1\cdd}{j}{\cdd n} \\
=~ (-1)^{j+1} \tih{1\cdd}{j}{\cdd n}~\! \Vih{2\cdd}{j}{\cdd n}~\! \Big(\prod_{\stackrel{k=2}{_{k\neq j}}}^{n-1}(a_k-a_1)\Big)~\!\big[(a_n-a_j) \tau_1 \tau_{j n} + (a_j-a_1) \tau_n \tau_{1 j}\big]\label{cofj1}
\end{multline}
in the linear combination \eqref{lc}. Clearly, if the Hirota-Miwa equation $E_{1 j n}$ 
$$(a_n-a_j) \tau_1 \tau_{j n} + (a_1-a_n) \tau_j \tau_{1 n} + (a_j-a_1) \tau_n \tau_{1 j} = 0~\!,$$
in the $\ell_1, \ell_j$ and $\ell_n$ directions is satisfied, then the combination \eqref{cofj1} is equal to :
\begin{equation}
(-1)^{j+1} \Vih{1\cdd}{j}{\cdd n}~\! \tau_j~\!\tau_{1 n}~\!\tih{1\cdd}{j}{\cdd n}\qquad (~\!^\forall j~\!:~2\leq j\leq n-1)~\!.\label{cofj2}
\end{equation}

It then follows from \eqref{cof1}, \eqref{cofj2} and \eqref{cofn} that the linear combination \eqref{lc} takes the form
$$\tau_{1 n}~\! \sum_{j=1}^{n} (-1)^{j+1}~ \Vih{1\cdd}{j}{\cdd n}~\!\tau_j~\!\tih{1\cdd}{j}{\cdd n}  ~\equiv~ \tau_{1 n}~\! E_{1\cdd n}~\!,$$
where $E_{1\cdd n}$ is identical to the l.h.s. of the $n$ term bilinear relation \eqref{dbKP}.

Hence, one finds that the following induction proves the proposition. \medskip

At $n=4$ : ~by taking the combination \eqref{lc}, the Hirota-Miwa equations $E_{1 2 3}$ and $E_{2 3 4}$ generate the relation $E_{1 2 3 4}$ (along the way, we require $E_{1 2 4}$ and $E_{1 3 4}$ to be satisfied as well, in order to be able to perform a reduction such as that from \eqref{cofj1} to \eqref{cofj2}).
\medskip

At $n=5$ : ~ As for $n=4$, the Hirota-Miwa equations $E_{1 2 3}, E_{1 2 4}, E_{1 3 4}$ and $E_{2 3 4}$ generate $E_{1 \cdd 4}$,
 whereas the set of Hirota-Miwa equations $\{E_{2 3 4}, E_{2 3 5}, E_{2 4 5}, E_{3 4 5}\}$ will generate $E_{2\cdd  5}$. The remaining $\DIS \dbinom{5}{3} - (4+3) = 3 $ equations $\{ E_{1 2 5}, E_{1 3 5}, E_{1 4 5}\}$ are required (in the reduction \eqref{cofj1}$\rightarrow$\eqref{cofj2}) to generate $E_{1\cdd 5}$ from $E_{1\cdd 4}$ and $E_{2\cdd 5}$.\medskip
 
 In general, the $\binom{n-1}{3}$  Hirota-Miwa equations that do not involve the $\ell_n$ direction are all needed to generate $E_{1\cdd n-1}$. The same number of equations involving $\ell_2, \hdots, \ell_n$ is required to construct $E_{2\cdd n}$ but among these there are $\binom{n-2}{3}$ equations that do  not involve $\ell_n$ and that were already used before. Hence, the number of extra equations required to construct $E_{2\cdd n}$ is $\binom{n-1}{3}-\binom{n-2}{3} = \binom{n-2}{2}$. Add to this the $n-2$ equations of the form $E_{1 j n}$ ($2\leq j\leq n$) needed to obtain $E_{1\cdd n}$ from $E_{1\cdd n-1}$ and $E_{2\cdd n}$ (in the reduction of coefficients like \eqref{cofj1} to \eqref{cofj2}), and we find that the total number of different Hirota-Miwa equations used in the construction is 
 \begin{gather*}
 \binom{n-1}{3} + \binom{n-2}{2} + (n-2) =  \binom{n-1}{3} + \binom{n-1}{2} = \binom{n}{3}~\!,
 \end{gather*}
 i.e., exactly the number of equations in the set $\mathbb{E}$.
 \end{proof}
 
 Moreover, it is easily seen that, as a generating set for the discrete KP hierarchy \eqref{dbKP}, $\mathbb{E}$ is not minimal. Singling out one particular direction (say, $\ell_1$, for simplicity) as special, one has the following :
 
 \begin{prop}
 The set of Hirota-Miwa equations $\mathbb{E}$ is generated by the $\binom{n-1}{2}$ equations in
$$\Big\{~\! E_{1 j k}~\!\Big|~ j,k,\in\{2, \hdots, n\},~ j<k~\!\Big\}~\!.$$
This set is (obviously) minimal.
 \end{prop}
 
 \begin{proof} In general, the $E_{1 j k}$ equation
 $$(a_k-a_j) \tau_1 \tau_{jk} + (a_1-a_k) \tau_j \tau_{1 k} + (a_j-a_1) \tau_k \tau_{1 j} = 0~\!,$$
 can be rewritten as:
 \begin{equation}
 \tau_1 = \frac{1}{a_k-a_j}~\! \frac{\tau_j \tau_k}{\tau_{jk}}~\! \big[ (a_k- a_1) \frac{\tau_{1 k}}{\tau_k} + (a_1-a_j) \frac{\tau_{1 j}}{\tau_j}\big]~\!, 
\label{almostlinear} 
 \end{equation}
 and similarly, from $E_{1 j \ell}$ and $E_{1 k \ell}$ ($j<k<\ell$), one has that
 \begin{align*}
  \tau_1 =& \frac{1}{a_\ell-a_j}~\! \frac{\tau_j \tau_\ell}{\tau_{j\ell}}~\! \big[ (a_\ell- a_1) \frac{\tau_{1 \ell}}{\tau_\ell} + (a_1-a_j) \frac{\tau_{1 j}}{\tau_j}\big]\\
 =& \frac{1}{a_\ell-a_k}~\! \frac{\tau_k \tau_\ell}{\tau_{k\ell}}~\! \big[ (a_\ell- a_1) \frac{\tau_{1 \ell}}{\tau_\ell} + (a_1-a_k) \frac{\tau_{1 k}}{\tau_k}\big]~,
 \end{align*}
from which one immediately finds that
 $$\tau_1 ~\!\big[ (a_\ell - a_k) \tau_j \tau_{k \ell} + (a_j - a_\ell) \tau_k \tau_{j \ell} + (a_k - a_j) \tau_\ell \tau_{j k}\big] = 0~\!.$$
 \end{proof}

Hence, it becomes clear that one could in fact choose any particular direction on the lattice and use the other (infinitely many!) directions to define `deformation' equations for the dependence of the eigenfunctions on the preferred direction, very much as is done in classical Sato theory. This is the approach taken in \cite{Glasgowpap}, to which the reader is referred for further details.

\begin{remark}
It is a well-known fact that the relation \eqref{almostlinear} generates a linear equation for the ratio $\tau_1/\tau$, which is -- up to a mere gauge transformation -- part of the adjoint linear problem \eqref{dKP-aLP} for the discrete KP hierarchy. Indeed, as was done in Lemma \ref{LPfromHM} for the linear problem, it suffices to define
$$\psi^* := \frac{\tau_1}{\tau} \prod_{s=2}^n \left(\frac{a_s-a_1}{a_s}\right)^{\ell_s}$$
to obtain
$$\psi^* = \frac{1}{a_k-a_j}~\! \frac{\tau_j \tau_k}{\tau \tau_{jk}}~\! \big[ a_k\psi^*_k -a_j \psi^*_j\big]~\!,$$
where, as for the tau functions, subscripts are used to denote shifts in the indicated direction for the function $\psi^*$. This fact is just one incarnation of the fundamental property of the discrete KP hierarchy, that a shift on the lattice is indistinguishable from a Darboux transformation
$$\tau~\mapsto~\tau~\! \psi^* \sim \tau_1$$
for the tau functions (defined, in this case, in terms of an adjoint eigenfunction $\psi^*$). This property, in fact, lies at the heart of extensions of the discrete KP theory to the non-commutative case \cite{Nimmo-noncom}.
\end{remark}



\begin{thebibliography}{99}
\bibitem{AbloL} Ablowitz, A.L. and J.F. Ladik, A nonlinear difference scheme and inverse scattering, Stud. Appl. Math. {\bf 55} (1976) 213-229.
\bibitem{Adlervmb} Adler, M. and P. van Moerbeke, A matrix integral solution to two-dimensional $W_p$-gravity, Comm. Math. Phys. {\bf 147} (1992) 25--56.
\bibitem{Brezin} Brezin, E. and V.A. Kazakov, Exactly solvable field theories of closed strings, Phys. Lett. B {\bf 236} (1990) 144--150.
\bibitem{BK} Broer, L. J. F., Approximate solutions for long water waves, Appl. Sci. Res. {\bf 31} (1975) 377--395.
\bibitem{Cheng} Cheng, Y., Constraints of the Kadomtsev-Petviashvili hierarchy, J. Math. Phys. {\bf 33} (1992) 3774--3782.
\bibitem{DateIII} Date, E., Jimbo, M. and T. Miwa, Method for generating discrete soliton equation: III,
J. Phys. Soc. Japan {\bf 52} (1983),  388--393.
\bibitem{DateIV} Date, E., Jimbo, M. and T. Miwa, Method for generating discrete soliton equation: IV,
J. Phys. Soc. Japan {\bf 52} (1983),  761--765.
\bibitem{Drinfeld} Drinfeld, V.G. and V. V. Sokolov: Lie algebras and equation of KdV type, J. Sov. Math. {\bf 30} (1985) 1975--2036.
\bibitem{FrenkelK} Frenkel, I.B. and V.G. Kac, Basic Representations of Affine Lie Algebras and Dual Resonance Models, Inv. math. {\bf 62} (1980) 23--66.
\bibitem{Gelfand} Gelfand, I. and L. Dikii, Asymptotic behaviour of the resolvent of Sturm-Liouville equations and the algebra of the Korteweg-de Vries equations, Russian Math. Surveys {\bf 30} (1975) 77--113.
\bibitem{Gari} Garifullin, R., Habibullin, I. and M. Yangubaeva, Affine and Finite Lie Algebras and Integrable Toda Field Equations on Discrete Space-Time, SIGMA {\bf 8} (2012) 062 (33p).
\bibitem{GramPRL} Grammaticos, B., Ramani, A. and V. Papageorgiou, Do integrable Mappings Have the Painlev\'e Property?, Phys. Rev. Lett. {\bf 67} (1991) 1825--1828.
\bibitem{GramCIMPA} Grammaticos, B. and R. Ramani, Discrete Painlev\'e Equations: A Review, in {\it Discrete Integrable Systems} (Lecture Notes in Physics 644), B. Grammaticos, Y. Kosmann-Schwarzbach, Tamizhmani, Tamizharasi (Eds.), Springer-Verlag Berlin Heidelberg (2004), p. 245--321.
\bibitem{Hata} Hatayama, G., Hikami, K., Inoue, R., Kuniba, A. and T. Tokihiro, The $A_N^{(1)}$ automata related to crystals of symmetric tensors, J. Math. Phys. {\bf 42} (2001) 274--308.
\bibitem{Haine} Haine, L. and P. Iliev, The bispectral property of a $q$-deformation of the Schur polynomials and the $q$-KdV hierarchy, J. Phys. A: Math. Gen. {\bf 30} (1997) 7217--7227.
\bibitem{Hattori} Hattori, M., Obtention of a discrete NLS equation through reduction from the Hirota-Miwa equation, Master Thesis, Graduate School of Mathematical Sciences, the University of Tokyo (2010). (in Japanese) 
\bibitem{Hiro-dKdV} Hirota, R., Nonlinear Partial Difference Equations. I. A Difference Analogue of the Korteweg-de Vries Equation, J. Phys. Soc. Jpn. {\bf 43} (1977) 1424--1433.
\bibitem{HirotaM} Hirota, R., Discrete analogue of a generalized Toda equation, J. Phys. Soc. Jpn. {\bf 50} (1981) 3785--3791.
\bibitem{IkedaY} Ikeda, T.  and H. Yamada, Polynomial $\tau$-Functions of the NLS-Toda Hierarchy and the Virasoro Singular Vectors, Lett. Math. Phys. {\bf 60} (2002) 147--156.
\bibitem{Ikedaetal} Ikeda, T., Mizukawa, H., Nakajima, T. and H. Yamada, Mixed expansion formula for the rectangular Schur functions and the affine Lie algebra $A_1^{(1)}$, Appl. Math. {\bf 40} (2008) 514--535.
\bibitem{JaulentM} Jaulent, M. and I. Miodek, Connection between Zakharov-Shabat and Schr\"odinger-type inverse-scattering transformations, Lett. Nuovo Cimento {\bf 20} (1977) 655--660.
\bibitem{JimboM} Jimbo, M. and T. Miwa, Solitons and infinite dimensional Lie algebras, Publ. RIMS. Kyoto Univ. {\bf 19} (1983) 943--1001.
\bibitem{Kaji-qP} Kajiwara, K., Noumi, N. and Y. Yamada, q-Painlev\'e systems arising from q-KP hierarchy, Lett. Math. Phys. {\bf 62} (2003) 259--268.
\bibitem{KakeiK} Kakei, S. and T. Kikuchi, A q-analogue of $\mathfrak{gl}_3$ hierarchy and q-Painlev\'e VI, J. Phys. A: Math. Gen. {\bf 39} (2006) 12179--12190.
\bibitem{Glasgowpap} Kakei, S., Nimmo, J.J.C. and R. Willox, Yang-Baxter maps and the discrete KP hierarchy, Glasgow Math. J. {\bf 51A} (2009) 107--119.
\bibitem{KonopelchenkoS} Konopelchenko, B. and W. Strampp, New reductions of the Kadomtsev-Petviashvili and two-dimensional Toda lattice hierarchies via symmetry constraints, J. Math. Phys. {\bf 33} (1992) 3676--3686.
\bibitem{Kricheveretal} Krichever, I., Lipan, O., Wiegmann, P. and A. Zabrodin, Quantum Integrable Models and Discrete Classical Hirota Equations, Comm. Math. Phys. {\bf 188} (1997) 267--304.
\bibitem{Lin} Lin, R., Liu, X. and Y. Zeng, A new extended q-deformed KP hierarchy, J. Nonl. Math. Phys. {\bf 15} (2008) 333--347.
\bibitem{Loris1} Loris, I. and R. Willox, Bilinear form and solutions of the $k$-constrained Kadomtsev-Petviashvili hierarchy,
Inverse Problems {\bf 13} (1997) 411--420.
\bibitem{MarunoK} Maruno, K. and K. Kajiwara, The discrete potential Boussinesq equation and its multi-soliton solutions, Applicable Analysis {\bf 89} (2010) 593-609.
\bibitem{Melnikov} Mel'nikov, V.K., Wave emission and absorption in a nonlinear integrable system, Phys. Lett. A {\bf 118} (1986) 22--24.
\bibitem{Miwa} Miwa, T., On Hirota's Difference Equations, Proc. Japan Acad. A {\bf 58} (1982) 9--12.
\bibitem{Mulase} Mulase, M., Algebraic Theory of the KP Equations, in {\it Perspectives in Mathematical Physics}, R. Penner and S. T. Yau (Eds.), International Press of Boston (1994), p. 151--218.
\bibitem{Nijhoff-PII} Nijhoff, F.W. and V.G. Papageorgiou, Similarity Reductions of Integrable Lattices and Discrete Analogues of Painlev\'e II Equation, Physics Letters A {\bf 153} (1991) 337--344.
\bibitem{Nijhoff-LGD} Nijhoff, F.W., Papageorgiou, V.G., Capel, H.W. and G.R.W. Quispel, The Lattice Gel'fand-Dikii Hierarchy, Inverse Problems {\bf8} (1992) 597--621.
\bibitem{Nijhoff-PVI} Nijhoff, F.W., Ramani, A., Grammaticos, B. and Y. Ohta, On Discrete Painlev\'e Equations associated with the Lattice KdV Systems and the Painlev\'e VI Equation, Studies in Applied Mathematics {\bf 106} (2001) 261--314.
\bibitem{Nimmo-DT} Nimmo, J.J.C., Darboux transformations and the discrete KP equation J. Phys. A. {\bf 30} (1997) 1--12.
\bibitem{Nimmo-noncom}Nimmo,  J. J. C., On a non-Abelian Hirota-Miwa equation, J. Phys. A {\bf 39} (2006) 5053-5065.
\bibitem{Oevel-etal} Oevel, W., Sidorenko, J. and W. Strampp, Hamiltonian structures of the Melnikov system and its reductions, Inverse Problems {\bf 9} (1993) 737--747.
\bibitem{Ohta-dKP} Ohta, Y., Hirota, R., Tsujimoto, S. and T. Imai, Casorati and Discrete Gram Type Determinant Representations of Solutions to the Discrete KP Hierarchy, J. Phys. Soc. Jpn. {\bf 62} (1993) 1872--1886.
\bibitem{Sachs} Sachs, R. L. , On the integrable variant of the Boussinesq system: Painlev\'e property, rational solutions, a related
many-body system, and equivalence with the AKNS hierarchy, Physica D {\bf 30} (1988) 1--27.
\bibitem{Sadakane} Sadakane, T., Ablowitz-Ladik hierarchy and two-component Toda lattice hierarchy. J. Phys. A: Math. Gen. {\bf 36} (2003) 87--97.
\bibitem{SidorenkoS} Sidorenko, J. and W. Strampp, Symmetry constraints of the KP hierarchy, Inverse Problems {\bf 7} (1991) L37--L43.
\bibitem{Suris} Suris, T.B. , A note on an integrable discretization of the nonlinear Schr\"odinger equation, Inverse Problems {\bf 13} (1997) 1121--1136.
\bibitem{Suzuki} Suzuki, T., A q-analogue of the Drinfeld-Sokolov hierarchy of type A and q-Painlev\'e system, arXiv:1105.4240v2 [math.QA]
\bibitem{TakahashiS} Takahashi, D. and J. Satsuma, A Soliton Cellular Automaton, J. Phys. Soc. Jpn. {59} (1990), 3514--3519.
\bibitem{Toki96} Tokihiro, T., Takahashi, D., Matukidaira, J. and J. Satsuma, From Soliton Equations to Integrable Cellular Automata through a Limiting Procedure, Phys. Rev. Lett. {\bf 76} (1996) 3247--3250.
\bibitem{TongasN} Tongas, A.S. and F.W. Nijhoff, The Boussinesq integrable system. Compatible lattice and continuum structures, Glasgow Math. J. {\bf 47A}  (2005) 205--219.
\bibitem{Ward} Ward, R.S., Discrete Toda field equations, Phys. Lett. A {\bf 199} (1995) 45-48.
\bibitem{RW-JMP1} Willox, R., Tokihiro, T. and J. Satsuma, Darboux and binary Darboux transformations for the nonautonomous
discrete KP equation, J. Math. Phys. {\bf 38} (1997) 6455--6469.
\bibitem{RW-IP} Willox, R., Tokihiro, T., Loris, I. and J. Satsuma,The fermionic approach to Darboux transformations, Inverse Problems {\bf 14} (1998) 745--762.
\bibitem{RW-JMP2} Willox, R. and I. Loris, KP constraints from reduced multi-component hierarchies, J. Math. Phys. {\bf 40} (1999) 6501--6525.
\bibitem{Wilson} Wilson, G., The Affine Lie Algebra $C_2^{(1)}$ and an Equation of Hirota and Satsuma, Phys. Lett. A {\bf 89} (1982) 332--334.
\bibitem{YO} Yajima, N. and M. Oikawa, Formation and interaction of sonic-Langmuir solitons -- inverse scattering method, Prog.
Theor. Phys. {\bf 56} (1976) 1719--1739.

\end{thebibliography}
\end{document}